\documentclass[letterpaper,journal]{IEEEtran}
\usepackage{amsmath,amsfonts,amssymb,amsthm}
\usepackage{algorithm}
\usepackage{algpseudocode}
\usepackage{array}
\usepackage{textcomp}
\usepackage{stfloats}
\usepackage{url}
\usepackage{verbatim}
\usepackage{graphicx}
\usepackage{wrapfig}
\usepackage{subfigure}
\usepackage{multirow}
\usepackage{booktabs}
\usepackage{cite}
\usepackage{lipsum}
\usepackage{mathrsfs}
\usepackage{enumitem}
\usepackage{bm}
\usepackage{bbm}
\usepackage{soul}
\usepackage{nicefrac}
\usepackage{microtype}
\usepackage{mathtools}
\usepackage{color}
\usepackage{xcolor}
\usepackage{hyperref}
\usepackage[normalem]{ulem}

\newtheorem{strategy}{Strategy}    
\newtheorem{theorem}{Theorem}

\newtheorem{definition}{Definition}

\begin{document}

\title{Efficient Mining of Low-Utility Sequential Patterns}

\author{Jian Zhu,~\IEEEmembership{Member,~IEEE}, Zhidong Lin, Wensheng Gan,~\IEEEmembership{Member,~IEEE}, Philip S. Yu,~\IEEEmembership{Life Fellow,~IEEE}

\thanks{This research was supported in part by National Natural Science Foundation of China (No. 62237001 and No. 62272196), Guangzhou Basic and Applied Basic Research Foundation (No. 2024A04J9971). (Corresponding author: Wensheng Gan)} 
    
\thanks{Jian Zhu and Zhidong Lin are with the School of Computer Science and Technology, Guangdong University of Technology, Guangzhou 510006, China. (E-mail: dr.zhuj@gmail.com, zhidonglin5@gmail.com, cairuichu@gdut.edu.cn)} 

\thanks{Wensheng Gan is with the College of Cyber Security, Jinan University, Guangzhou 510632, China. (E-mail: wsgan001@gmail.com)} 

\thanks{Philip S. Yu is with the Department of Computer Science, University of Illinois Chicago, Chicago, USA. (E-mail: psyu@uic.edu)} 
}
\markboth{}%
{Zhu \MakeLowercase{\textit{et al.}}: Efficient Mining of Low-Utility Sequential Patterns}

\maketitle

\begin{abstract}
    Discovering valuable insights from rich data is a crucial task for exploratory data analysis. Sequential pattern mining (SPM) has found widespread applications across various domains. In recent years, low-utility pattern mining has shown strong potential in applications such as intrusion detection and genomic sequence analysis. However, existing utility-based sequential pattern mining studies mainly focus on high-utility sequential patterns, and the definitions and strategies used in high-utility sequential pattern mining cannot be directly applied to low-utility sequential pattern mining (LUSPM). Moreover, there is currently no algorithm specifically designed for mining low-utility sequential patterns. To this end, we formalize the LUSPM problem, redefine sequence utility, and introduce a compact data structure called the sequence-utility chain to efficiently record utility information. Furthermore, we propose three novel algorithms—LUSPM$_{b}$, LUSPM$_{s}$, and LUSPM$_{e}$—to discover the complete set of low-utility sequential patterns. LUSPM$_{b}$ serves as an exhaustive baseline, while LUSPM$_{s}$ and LUSPM$_{e}$ build upon it, generating subsequences through shrinkage and extension operations, respectively. In addition, we introduce the maximal non-mutually contained sequence set and incorporate multiple pruning strategies, which significantly reduce redundant operations in both LUSPM$_{s}$ and LUSPM$_{e}$. Finally, extensive experimental results demonstrate that both LUSPM$_{s}$ and LUSPM$_{e}$ substantially outperform LUSPM$_{b}$ and exhibit excellent scalability. Notably, LUSPM$_{e}$ achieves superior efficiency, requiring less runtime and memory consumption than LUSPM$_{s}$. Our code is available at https://github.com/Zhidong-Lin/LUSPM.
\end{abstract}


\section{Introduction}

\IEEEPARstart{W}ith the advent of big data, the demand for processing large-scale datasets and extracting valuable knowledge has increased significantly. Data mining and analytics \cite{chen2002data} have emerged as crucial technologies for uncovering essential knowledge from diverse data sources. Pattern mining \cite{han2001prefixspan,chen2024towards} has been widely applied to identify meaningful patterns, including itemsets \cite{li2022frequent,tung2025mining,chen2025toward}, sequences \cite{agrawal1995mining,qiu2021efficient,zhang2021tkus,gan2019survey}, and rules \cite{gan2023anomaly,zhu2024targeted}. Among these, early studies on sequential pattern mining (SPM) focused primarily on the sequence's frequency. However, frequency alone may overlook other important aspects, motivating the development of utility-based SPM \cite{yin2012uspan,gan2020proum,wang2018incremental}. Although numerous utility-based SPM algorithms have been proposed, existing research has largely focused on high-utility sequential pattern mining (HUSPM) \cite{shie2011mining,lan2014applying,alkan2015crom}, while low-utility sequential pattern mining (LUSPM) remains underexplored. It is important to note that low utility does not imply low importance. In practice, utility reflects quantitative contribution (e.g., economic profit or accumulated value) rather than informational significance. Low-utility patterns often correspond to rare or subtle behaviors that deviate from dominant patterns and may indicate anomalies, security threats, system faults, or abnormal biological processes. For example, in intrusion detection, individual malicious login attempts may have low utility but form suspicious sequences in specific orders; in genomic analysis, low-expression DNA/RNA subsequences may reveal abnormal regulatory mechanisms. Such patterns are typically filtered out by threshold-based HUSPM methods, highlighting the need for dedicated LUSPM techniques. Despite its practical and theoretical importance, to the best of our knowledge, no dedicated algorithms have been developed for LUSPM. Therefore, this paper presents the first systematic study of LUSPM, formally defining the problem and proposing effective techniques to discover low-utility yet information-rich sequential patterns. Nevertheless, this task poses several fundamental challenges.

\begin{table}[t!]
	\centering
    \scriptsize
	\caption{A quantitative sequence database}
	\label{table2}
	\begin{tabular}{|c|c|}  
		\hline 
		$S_{id}$ & \textbf{q-Sequence}\\		\hline  
		$S_{1}$& $\langle$($a$: 1), ($b$: 2), ($c$: 1), ($a$: 2), ($b$: 3), ($d$: 3), ($a$: 3)$\rangle$\\ 
		\hline
		$S_{2}$& $\langle$($g$: 1), ($a$: 1), ($b$: 2), ($c$: 2), ($d$: 1)$\rangle$\\  
		\hline  
		$S_{3}$& $\langle$($e$: 2), ($f$: 2), ($a$: 1), ($b$: 2), ($e$: 2)$\rangle$\\
		\hline  
		$S_{4}$& $\langle$($f$: 1), ($e$: 2), ($a$: 2), ($b$: 2), ($a$: 2), ($b$: 2)$\rangle$\\
		\hline
		$S_{5}$& $\langle$($d$: 2), ($a$: 1), ($c$: 3), ($d$: 2)$\rangle$\\ 
		\hline
 $S_{6}$&$\langle$($c$: 1), ($a$: 2), ($d$: 3), ($a$: 3)$\rangle$\\\hline
	\end{tabular}
\end{table}

At first, the conventional sequence utility definition in HUSPM is not well suited for LUSPM. In most HUSPM algorithms, the utility of a sequence is typically defined as the maximum, minimum, or average utility across transaction sequences~\cite{zhang2023husp,truong2020ehusm,truong2020ehausm}, whereas LUSPM requires the total utility over all occurrences. This difference may lead to misleading evaluations because the conventional definition captures only partial utility information. For example, assume that the utility threshold is 6 and that the internal utility equals the actual item utility. In Table~\ref{table2}, the sequence $S$ = $\langle$$d$, $a$$\rangle$ occurs once in $S_1$ with utility 3 + 3 = 6, yielding $u$($S$, $S_1$) = $max$(6) = 6 $=$ 6, and thus it is considered high-utility. By contrast, $Q$ = $\langle$$a$, $b$$\rangle$ appears three times in $S_1$ with utilities 1 + 2 = 3, 1 + 3 = 4, and 2 + 3 = 5, giving $u$($Q$, $S_1$) = $max$(3, 4, 5) = 5 $<$ 6, and therefore it is not identified as high-utility. However, the total utility of $S$ is 6, whereas that of $Q$ is 12, which is substantially larger. This example shows that the conventional HUSPM utility definition captures only partial utility and may overlook sequences with greater overall contribution.

Moreover, LUSPM faces challenges in computational efficiency and memory consumption. On the one hand, similar to HUSPM, calculating sequence utilities requires comprehensive information from the database, which substantially increases both computational and memory costs. On the other hand, the discovery process generates a substantial number of candidate sequences. Existing pruning strategies in HUSPM are designed to preserve sequences whose utilities exceed a given threshold by leveraging utility upper bounds. However, since LUSPM aims to discover sequences whose utilities fall below the threshold, these properties no longer support effective pruning, rendering the existing strategies inapplicable or even unsafe in this setting.

To address these challenges and improve the efficiency of LUSPM, this paper first investigates the feasibility of applying existing HUSPM algorithms to mine low-utility sequences. The results show that this approach is infeasible (see Section \ref{Investigation of Mining LUSPs Using HUSPM Methods}). To overcome this limitation, we redefine sequence utility to enable the accurate discovery of low-utility sequential patterns. Specifically, sequence utility is defined as the sum of utilities across all transactional sequences, thereby reflecting the true utility of a sequence in the database. Based on this definition, we propose a simple algorithm, LUSPM$_{b}$, to discover the complete set of LUSPs. In particular, LUSPM$_{b}$ adopts an exhaustive search strategy to identify LUSPs and introduces a novel data structure called the sequence-utility (SU) chain to precisely capture sequence utility information. However,LUSPM$_{b}$ suffers from high computational cost and low efficiency.

In order to address this problem, we propose two improved algorithms, LUSPM$_{s}$ and LUSPM$_{e}$, to more effectively mine LUSPs. LUSPM$_{s}$ is a shrinkage-based algorithm that derives shorter sequences by progressively removing items from longer sequences, whereas LUSPM$_{e}$ is an extension-based algorithm that constructs longer sequences by inserting items into shorter ones. To reduce redundant operations, both algorithms introduce the maximal non-mutually contained sequence set (MaxNonConSeqSet) to prune invalid sequences. In addition, we propose four pruning strategies called EUPS, SLUSPS, SBIPS, and EBISPS to improve efficiency. Specifically, EUPS is applied during preprocessing to eliminate invalid items in the MaxNonConSeqSet. SLUSPS and SBIPS are employed in LUSPM$_{s}$ to avoid unnecessary utility computations and prune invalid items, while EBISPS is used in LUSPM$_{e}$ to prune a substantial number of invalid sequences. Overall, the integration of these pruning strategies enables both algorithms to efficiently mine LUSPs. The key contributions of this paper are as follows:

\begin{itemize}
    \item To address the task of discovering low-utility yet informative sequential patterns, we introduce the concept of LUSPs, redefine the sequence utility, and formalize the LUSPM problem. To our knowledge, this is the first study focusing on LUSPM.
    \item We develop a basic algorithm called LUSPM$_{b}$ that utilizes a structure called sequence-utility chain to capture the utility information and is capable of mining the complete set of LUSPs.
    \item Building on LUSPM$_{b}$, we develop two improved algorithms, LUSPM$_{s}$ and LUSPM$_{e}$, which leverage MaxNonConSeqSet and four pruning strategies to significantly enhance the mining efficiency of LUSPs.
    \item We conduct extensive experiments on six datasets, and the results show that both LUSPM$_s$ and LUSPM$_e$ significantly outperform LUSPM$_{b}$ with great scalability, and LUSPM$_e$ achieves superior runtime and memory efficiency compared with LUSPM$_{s}$.
\end{itemize}

The paper is organized as follows: Section \ref{sec: relatedwork} reviews related work; Section \ref{sec: preliminaries} introduces basic concepts and problem definition; Section \ref{Investigation of Mining LUSPs Using HUSPM Methods} investigates the feasibility of using HUSPM methods to mine LUSPs; Section \ref{sec: algorithm} presents the proposed algorithms; Section \ref{sec: experiments} shows experimental results; Section \ref{sec: conclusion} provides conclusions and future research directions. 

\section{Related Work}  \label{sec: relatedwork}
\subsection{Frequent Sequential Pattern Mining} \label{subsec:fspm}

As a key part of exploratory data analysis, pattern mining extracts meaningful patterns, such as itemsets, sequences, and rules, from databases \cite{fournier2022pattern}. Among them, sequences specifically capture the temporal order between items. Sequential pattern mining (SPM) was first proposed to identify useful sequential patterns \cite{agrawal1995mining}, which can be used in customer shopping, traffic, web access, stock trends, and DNA analysis \cite{han2001prefixspan,srikant1996mining}.  Since then, many algorithms have been proposed to discover frequent sequential patterns. The early well-known SPM algorithm was AprioriAll \cite{agrawal1995mining}, which relied on the Apriori property of frequent sequences. However, it faced efficiency issues when handling large-scale data. To improve efficiency, FreeSpan \cite{han2000freespan} introduced the projected sequence database to constrain subsequence exploration and reduce candidate generation. However, generating projected databases incurred high costs. PrefixSpan \cite{han2001prefixspan} improved efficiency by recursively using frequent sequences as prefixes and projecting databases to narrow the search space. Additionally, SPM, which uses a bitmap representation (SPAM) \cite{ayres2002sequential}, was proposed for support calculation, thereby reducing memory usage. Finally, based on SPAM, CM-SPAM \cite{fournier2014fast} incorporated CMAP and co-occurrence pruning to further enhance performance. 

Traditional SPM algorithms often generate numerous less meaningful sequences. To address this issue, more advanced algorithms have been developed. Among them, closed SPM algorithms \cite{wu2020netncsp,fournier2014fast,fumarola2016clofast} and maximal SPM \cite{li2022netnmsp,fournier2013mining} reduce the number of mined frequent patterns through pattern compression. In addition, top-k sequential pattern mining (TSPM) \cite{petitjean2016skopus,fournier2013tks} and targeted sequential pattern mining (TaSPM) \cite{chiang2003goal,hu2024targeted} can reduce the number of sequences based on user requirements. Specifically, TSPM identifies the top-k frequent patterns that satisfy user-defined constraints, whereas TaSPM extracts sequences containing a user-specified target frequent sequence \cite{chiang2003goal,hu2024targeted}. However, frequency is not always a sufficient measure of pattern interestingness, since it disregards key aspects such as profitability, cost, and risk. This has led to the emergence of utility-based SPM \cite{gan2021survey}.

\subsection{High-Utility Sequential Pattern Mining} \label{subsec:huspm}

Traditional utility-based algorithms focus on high-utility sequential patterns (HUSPs), designed to identify sequences with utilities exceeding a given threshold. Unlike frequent-based SPM, high-utility SPM doesn't possess the Apriori property, which results in a huge search space and low efficiency. Initially, Shie \textit{et al.} \cite{shie2011mining} proposed the UMSP and UM-span algorithms to meet the needs of mobile business applications. Then the formalization of HUSPM was presented \cite{yin2012uspan}, alongside their efficient USpan algorithm \cite{yin2012uspan}. Utilizing pruning strategies, USpan reduced the search space and improved efficiency. Despite this, USpan cannot discover the complete set of HUSPs. To address this limitation, Lan \textit{et al.} \cite{lan2014applying} proposed a sequence-utility upper-bound, which can discover complete HUSPs. Then, HuspExt \cite{alkan2015crom} obtained a smaller upper bound by calculating the cumulated rest of match to narrow the search space. Wang \textit{et al.} \cite{wang2016efficiently} then proposed HUS-Span, which introduces two upper bounds—PEUs and RSUs—to further reduce unpromising candidates. However, the problem of large candidate sets still exists. Consequently, the projection-based ProUM method was proposed, which efficiently mines HUSPs based on a utility-list structure \cite{gan2020proum}. However, it remains insufficiently compact, and the pruning strategies employed are not sufficiently robust. Therefore, the HUSP-ULL algorithm \cite{gan2020fast} was proposed using the UL-list to discover HUSPs more efficiently. The previously mentioned algorithms are still restricted by memory usage limitations. Then, the HUSP-SP algorithm \cite{zhang2023husp} proposed a new utility upper bound called TRSU and significantly reduced the number of candidate patterns. 

In addition to these common high-utility sequential patterns, the TUS algorithm \cite{yin2013efficiently} focused on mining the top-k sequences based on user requirements. IncUSP-Miner+  \cite{wang2018incremental} was proposed to discover HUSPs incrementally. In addition, the previously mentioned algorithms all calculated utility under an optimistic scenario, i.e., the utility of a sequence was defined as the sum of the maximum utilities among all its occurrences, which may overestimate the actual utility of the pattern. Truong \textit{et al.}  \cite{truong2020ehusm} proposed utility calculation under a pessimistic scenario, where the utility of a sequence was defined as the sum of the minimum utilities among all its occurrences. Besides, the high average-utility sequence mining problem was also formulated \cite{truong2020ehausm}. However, these utility calculation strategies only capture partial utility information of a sequence and do not consider the complete utility information across all its occurrences. Moreover, while numerous studies have addressed high-utility or frequent sequential patterns, none have systematically investigated the LUSPM task. 

\subsection{Low-Utility Itemset Mining} \label{subsec:luipm}
While most utility mining research focuses on high-utility patterns, there is growing interest in low-utility patterns for anomaly detection. Low-utility itemset mining (LUIM) helps identify abnormal patterns, making it valuable in retail, healthcare, and fraud detection. However, upper-bound pruning strategies from high-utility pattern mining are unsuitable for low-utility pattern mining, as they may eliminate meaningful low-utility patterns. In 2019, Alhusaini \textit{et al.} \cite{alhusaini2019luim} first formulated the LUIM problem and proposed two algorithms: LUG-Miner and LUIMA. Specifically, LUG-Miner extracts high-utility generators and low-utility generators (LUGs), and LUIMA obtains LUIs using LUGs. However, the two algorithms could not discover complete results. Zhang \textit{et al.} \cite{zhang2025enabling} proposed the LUIMiner algorithm, which incorporates two lower bounds and pruning strategies to drastically narrow the search space of LUIM and redesigned a search tree to reorganize the traversal logic of LUIM. Despite existing advances, current studies are limited to low-utility itemset mining, with no work on low-utility sequential pattern mining. Unlike itemsets, sequential patterns impose ordering constraints and more complex matching, resulting in higher computational complexity and an exponentially growing search space. Thus, LUIM techniques are not directly applicable to LUSPM, where candidate explosion and computational cost are more severe. Therefore, dedicated algorithms are required for efficient low-utility sequential pattern mining.
\section{Preliminaries}  \label{sec: preliminaries}

In this section, we present the fundamental definitions and concepts related to the low-utility sequential pattern mining (LUSPM) problem.

\subsection{Concepts and Definitions}

Let $I$ = \{$i_1$, $i_2$, $\ldots$, $i_m$\} be a set of distinct items. A sequence $T$ = $\langle$$i_1$, $i_2$, $\ldots$, $i_n$$\rangle$ is an ordered list of items. A quantitative item ($q$-item) is denoted as ($i_k$: $q_k$), where $q_k$ represents its internal utility. A quantitative sequence ($q$-sequence) $S$ is an ordered list of $q$-items: $\langle$($i_1$: $q_1$), ($i_2$: $q_2$), $\ldots$, ($i_n$: $q_n$)$\rangle$. A quantitative sequential database $\mathcal{D}$ = \{$S_1$, $S_2$, $\ldots$, $S_n$\} contains multiple $q$-sequences, each with a unique identifier ${S_{id}}$. Each item also has an external utility, denoted as $ex(i_j)$. For example, consider the database $\mathcal{D}$ = \{$S_{1}$, $S_{2}$, $\ldots$, $S_{6}$\}, which contains seven items (i.e., $I$ = \{$a$, $b$, $c$, $d$, $e$, $f$, $g$\}) as shown in Table \ref{table2}. Each item is associated with an external utility, as listed in Table \ref{table3}.

\vspace{-20px}
\begin{table}[H]
	\centering
	\scriptsize
\caption{An external utility table}
\label{table3}
	\begin{tabular}{|>{\centering\arraybackslash}p{0.2\linewidth}|c|c|c|c|c|c|c|}  
		\hline 
		\textbf{item}& a &b & c& d &e& f&g\\
		\hline  
		\textbf{external utility}& 1 &3 & 1& 2 &1& 3&3\\\hline
	\end{tabular}
\end{table}

\begin{definition}[Q-sequence Containing \cite{yin2012uspan}]
    \rm A q-sequence $Q$ is a subsequence of $S$ (denoted as $Q$ $\subseteq$ $S$) if all q-items in $Q$ appear in $S$ in the same order; conversely, if $S$ is a subsequence of $Q$, then $S$ is also called a super-sequence of $Q$ (denoted as $S$ $\subseteq$ $Q$).
\end{definition}

\begin{definition}[Matching \cite{yin2012uspan}]\label{def:matching}
    \rm A sequence $T$ matches a q-sequence $Q$ (denoted as $T$ $\sim$ $Q$) if they contain the same items in the same order; a sequence $T$ may correspond to multiple $Q$.
\end{definition}

For example, let $S_6$ = $\langle$($c$: 1), ($a$: 2), ($d$: 3), ($a$: 3)$\rangle$ and $Q_1$ = $\langle$($c$: 1), ($a$: 2)$\rangle$. Then $Q$ $\subseteq$ $S_6$ and $S_6$ is a super-sequence of $Q_1$. Let $T$ = $\langle$$a$, $b$, $c$, $a$$\rangle$. Then $T$ matches $Q_2$ = $\langle$($a$: 1), ($b$: 2), ($c$: 1), ($a$: 2)$\rangle$, but does not match $Q_3$ = $\langle$($a$: 1), ($b$: 2), ($c$: 2)$\rangle$.

\begin{definition}[Length of Sequence \cite{yin2012uspan}]
    \label{def:length}
    \rm For a $q$-sequence $S$ = $\langle$($i_1$: $q_1$), ($i_2$: $q_2$), $\ldots$, ($i_n$: $q_n$)$\rangle$, its length $|S|$ is defined as the number of $q$-items it contains, which is $n$. 
\end{definition}

For example, for the sequences $Q_2$ and $Q_3$ in Definition \ref{def:matching}, their lengths are $|Q_2|$ = 4 and $|Q_3|$ = 3.

\begin{definition}[Support of Sequence \cite{agrawal1995mining}]
    \rm Given a sequence $T$ = $\langle$$j_{1}$, $j_{2}$, $\ldots$, $j_{m}$$\rangle$, the support of sequence $T$, denoted as \textit{sup}$(T)$, is the number of times $T$ appears in the sequence database.
\end{definition}

For example, in Table \ref{table2}, the sequence $T$ = $\langle$$a$, ${b}$, ${c}$$\rangle$ appears once in both $S_{1}$ and  $S_{2}$. Hence, \textit{sup}$(T)$ = \textit{sup}($\langle$${a}$, ${b}$, ${c}$$\rangle$) = 2.

\begin{definition}[Utility of Q-item]
    \rm The utility of a $q$-item ($i$: $p$) at position $j$ - 1 in $S$ is defined as:
        \begin{flalign}
    &&
       u(i, j - 1, S) = q(i, j - 1, S) \times {ex}(i).
    &&
    \end{flalign}
\end{definition}

For example, referring to Table \ref{table3}, consider $S_6$ = $\langle$($c$: 1), ($a$: 2), ($d$: 1), ($a$: 3)$\rangle$. The utility of the $q$-item ($a$: 2) at the second position of $S_6$ is calculated as $u$($a$, 1, $S_6$) = 2 $\times$ 1 = 2.

\begin{definition}[Utility of Sequence \cite{yin2012uspan}]
    \label{Utility of sequence}
    \rm Consider a $q$-sequence  $S$ in the $q$-database $\mathcal{D}$ = \{$S_1$, $S_2$, $\ldots$, $S_n$\} and its subsequence $Q$. The utility of $Q$ in $S$ is defined as:
    \begin{flalign}
    &&
        u(Q, S) = \sum_{0 \leq k \leq m-1} u(j_{k}, k, Q).
    &&
    \end{flalign}
    Where $m$ is the length of $Q$. When $S$ and $Q$ are identical, we denote $u(S)$ = $u$($S$, $S$) = $u$($Q$, $S$). The utility of a sequence $T$ in $\mathcal{D}$, denoted as $u(T)$, is defined as:
\begin{flalign}
&&
u(T) = \sum_{S_i \in \mathcal{D}} \sum_{\substack{Q \subseteq S_i \ T \sim Q}} u(Q, S_i),
&&
\end{flalign}
where $Q$ represents any $q$-subsequence occurrence of $T$ in sequence $S_i$, and $u(Q, S_i)$ denotes the utility of $Q$ in $S_i$.  If $T$ appears multiple times in $S_i$, each occurrence contributes to the utility calculation.
\end{definition}

Let us consider the $q$-database $\mathcal{D}$ = \{$S_{1}$, $S_{2}$, $\ldots$, $S_{n}$\} shown in Table \ref{table2}, sequence $T$ = $\langle$$a$, $b$$\rangle$ appears in $S_{1}$, $S_{3}$ and $S_{4}$. Thus, the utility of $T$ in $\mathcal{D}$ is calculated below: $u(T)$ = $ \sum_{S_{i} \subseteq D}$ $u$($T$, $S_{i}$) = $u$($T$, $S_{1}$) + $u$($T$, $S_{3}$) + $u$($T$, $S_{4}$) = $u$($\langle$($a$: 1), ($b$: 2)$\rangle$, $S_{1}$) + $u$($\langle$($a$: 1), ($b$: 3)$\rangle$, $S_{1}$) + $u$($\langle$($a$: 2), ($b$: 3)$\rangle$, $S_{1}$) + $u$($\langle$($a$: 1), ($b$: 2)$\rangle$, $S_{2}$) + $u$($\langle$($a$: 1), ($b$: 2)$\rangle$, $S_{3}$) + $u$($\langle$($a$: 2), ($b$: 2)$\rangle$, $S_{4}$) + $u$($\langle$($a$: 2), ($b$: 2)$\rangle$, $S_{4}$) + $u$($\langle$($a$: 2), ($b$: 2)$\rangle$, $S_{4}$) = 66. 

\begin{definition}[Utility of Database \cite{yin2012uspan}]
    \rm The utility of the $q$-database $\mathcal{D}$ = \{$S_{1}$, $S_{2}$, $\ldots$, $S_{n}$\} is defined as:
    \begin{flalign}
    &&
   u(\mathcal{D}) = \sum_{1 \leq i \leq n} u( S_{i}).
    &&
    \end{flalign}
\end{definition}

\begin{definition}[Low-utility Sequential Pattern]
    \label{def:minimum utility}
    \rm A sequence $T$ is called a low-utility sequential pattern (LUSP) in a $q$-database 
    $\mathcal{D} = \{S_{1}, S_{2}, \ldots, S_{n}\}$ if it satisfies $u(T)$ $\leq$ \textit{minUtil}, where $\textit{minUtil} = \sigma \times u(\mathcal{D})$ and $\sigma$ is the minimum utility threshold.
\end{definition}

\subsection{Problem Formulation}

Given a $q$-sequence database $\mathcal{D}$, a minimum utility threshold $minUtil$ (Definition~\ref{def:minimum utility}), and a maximum length constraint $maxLen$ (Definition~\ref{def:length}), the task of LUSPM is to discover the complete set of LUSPs that satisfy these conditions.

\section{Investigation of Mining LUSPs with HUSPM Methods} \label{Investigation of Mining LUSPs Using HUSPM Methods}

In this section, we investigate the feasibility of mining low-utility sequential patterns using existing HUSPM algorithms. Existing HUSPM algorithms, such as USpan \cite{yin2012uspan} and HUSP-SP \cite{zhang2023husp}, are designed to discover sequences that satisfy a minimum utility threshold and are therefore not directly applicable to low-utility mining. Specifically, HUSPM focuses on sequences satisfying $u$($S$) $\geq$ $minUtil$, whereas low-utility mining targets sequences with 0 $<$ $u$($S$) $\leq$ $minUtil$, resulting in fundamentally different search directions. Moreover, HUSPM relies on upper-bound-based pruning strategies: if the maximum possible utility of a sequence is smaller than the threshold, the sequence and its extensions can be safely pruned. However, in low-utility mining, small utilities are precisely the target, rendering such pruning strategies ineffective and limiting the ability to reduce the search space. To address this issue, one may attempt to solve it through the following problem transformation strategies.

\textbf{Utility inversion method}. A straightforward idea is to transform low-utility mining into high-utility mining via utility inversion. Specifically, for each item $i$, its utility is converted as $u'$($i$) = $u_{max}$ - $u(i)$, where $u_{max}$ is the maximum item utility in the database. After this transformation, low-utility items become relatively high-utility, enabling the use of existing HUSPM algorithms. However, this idea has limitations. Utility inversion may be affected by sequence length, altering the utility distribution within sequences and breaking the original relative utility relationships. As a result, low-utility sequences in the original database may not correspond to high-utility patterns after transformation, and thus correctness cannot be guaranteed.

\textbf{Complement-based method}. Another intuitive idea is to first mine all high-utility sequential patterns and then derive low-utility patterns by taking the complement over all possible sequences. However, this approach is impractical for two reasons. First, high-utility and low-utility patterns have fundamentally different utility definitions, so the complement of high-utility patterns does not necessarily correspond to the true set of low-utility patterns, and thus cannot guarantee completeness or correctness. Second, when the minimum utility threshold is small, the number of high-utility patterns may grow explosively, making complement computation over a massive pattern set computationally prohibitive.

In summary, the above ideas are generally ineffective for solving the LUSPM problem. On the one hand, the pruning mechanisms of existing HUSPM algorithms are not compatible with the search objective of low-utility mining. Therefore, directly applying HUSPM algorithms to mine LUSPs is generally infeasible. On the other hand, problem transformations may alter the original utility relationships among sequences, thereby affecting the correctness of the mining results. Consequently, it is necessary to design specialized algorithms for mining low-utility sequential patterns.

\section{Algorithm Design} \label{sec: algorithm}
In this section, we present three algorithms, namely LUSPM$_b$, LUSPM$_s$, and LUSPM$_e$, to address the LUSPM problem. We first introduce the shared data structures of these algorithms, followed by a description of LUSPM$_b$. Then, we provide detailed explanations of the search tree, pruning strategies, and procedures adopted by LUSPM$_s$ and LUSPM$_e$. Finally, we analyze the time and space complexities of LUSPM$_s$ and LUSPM$_e$.

\subsection{Data Structure} \label{data structure}

We first describe the key data structures employed in the three algorithms: a bit matrix \cite{ayres2002sequential} for efficient sequence presence verification and a sequence chain for utility recording.

\subsubsection{Bit Matrix} \label{bit matrix}

To efficiently verify sequence presence, we employ a bitmap data structure \cite{ayres2002sequential}, which encodes each item as a binary vector indicating its presence (1) or absence (0) in the sequence. For example, for $S$ = $\langle$$a$, $b$, $c$, $a$, $b$, $d$, $a$$\rangle$, the bit matrix of item $a$ is (1, 0, 0, 1, 0, 0, 1). This representation enables rapid presence checks using simple bitwise operations, thereby reducing computational cost.

\subsubsection{Sequence-Utility Chain} \label{sequence utilities chain}

To enhance utility computation, we propose a sequence-utility (SU) chain structure for storing sequence utility information. It consists of a set of nodes, where each node represents the utility of a sequence in a specific occurrence within the database. For a sequence $S$ = $\langle$$i_{1}$, $i_{2}$, $\ldots$, $i_{n}$$\rangle$, its sequence-utility chain is defined as $M$ = $\langle$$\langle$$a_{11}$, $a_{12}$, $\ldots$, $a_{1n}$$\rangle$, $\langle$$a_{21}$, $a_{22}$, $\ldots$, $a_{2n}$$\rangle$, $\ldots$, $\langle$$a_{m1}$, $a_{m2}$, $\ldots$, $a_{mn}$$\rangle$$\rangle$, where $m$ is the number of occurrences of $S$ in the database, and $a_{pq}$ denotes the utility of the $q$-th item of $S$ in its $p$-th occurrence. For example, Fig. \ref{fig:Sequence-utility Chain} illustrates the sequence-utility chain corresponding to the sequences $\langle$$a$, $b$, $c$$\rangle$, $\langle$$a$, $b$$\rangle$, $\langle$$a$$\rangle$, and $\langle$$g$, $a$$\rangle$ from Table \ref{table2}. Since the sequence $Q$ = $\langle$$a$, $b$, $c$$\rangle$ appears twice, once in $q$-sequence $S_{1}$ and once in  $S_{2}$, the corresponding sequence chain $N$ for $Q$ is $\langle$$\langle$1, 2, 1$\rangle$, $\langle$1, 2, 2$\rangle$$\rangle$. Here, we use internal utility for simplicity, while the actual utility is obtained by multiplying internal utility by external utility. This compact design not only reduces memory consumption but also streamlines utility computation, thereby improving both efficiency and scalability. 

\begin{figure}[htbp]
    \centering
    \includegraphics[width=\linewidth]{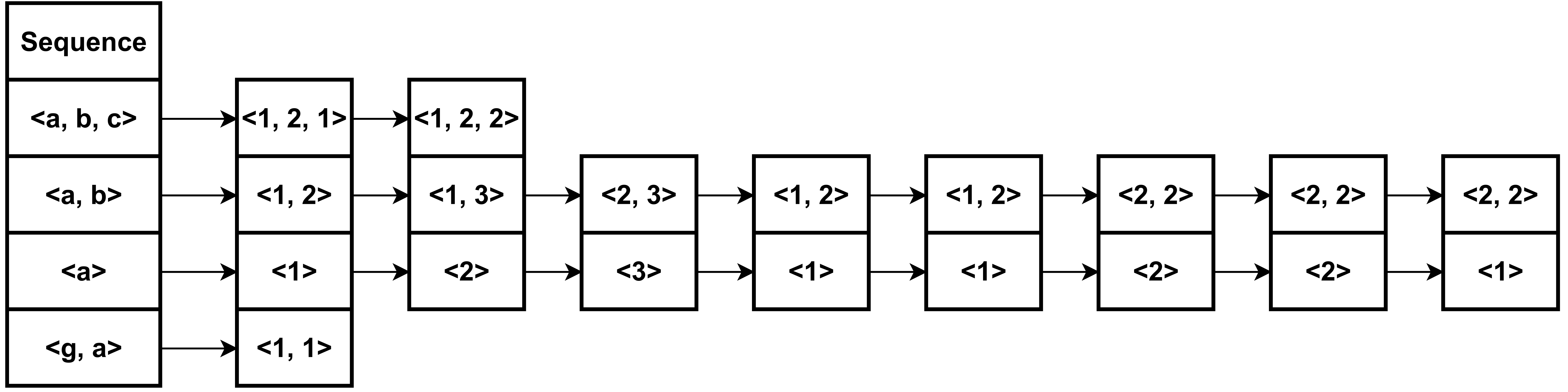}
    \caption{Sequence-utility chain of sequences$\langle$$a$, $b$, $c$$\rangle$, $\langle$$a$, $b$$\rangle$, $\langle$$a$$\rangle$, and $\langle$$g$, $a$$\rangle$.}
    \label{fig:Sequence-utility Chain}
\end{figure}
\vspace{-10px}
\subsection{LUSPM$_b$} \label{basic methods}

To discover the complete set of low-utility sequential patterns (LUSPs), we propose LUSPM$_b$ to discover LUSPs using exhaustive enumeration, and its procedure is presented in Algorithm~\ref{LUSPM_naive}. Specifically, LUSPM$_b$ begins by enumerating all possible candidate sequences from the database (line 1). For each candidate sequence $S$, the \textit{computeUtility} method is invoked to calculate its utility, after which the algorithm checks whether the utility exceeds \textit{minUtil} and whether the length of $S$ is no greater than \textit{maxLen} to determine if $S$ can be identified as a LUSP (lines 2–6). Although this method guarantees completeness, it relies on exhaustive search, which incurs extremely high computational costs. Moreover, when applied to large-scale data, the number of candidate sequences grows exponentially, rendering this enumeration approach infeasible in practice. Therefore, it is essential to design more effective strategies to improve the efficiency of LUSPM.

\begin{algorithm}[!h]
    \caption{LUSPM$_b$}
    \label{LUSPM_naive}
    \small
    \begin{algorithmic}[1]
    \Statex \textbf{Input:} $\mathcal{D}$: a sequence database; \textit{minUtil}: the utility threshold; \textit{maxLen}: the length restriction.
    \Statex \textbf{Output:} \textit{LUSPs}: the complete set of LUSPs.

    \State scan $\mathcal{D}$ to generate \textit{allSequenceSet};
    \For{each $S$ $\in$ \textit{allSequenceSet}}
    \If{\textit{computeUtility(S)} $\leq$ \textit{minUtil} \textbf{and}  $|S|$ $\leq$ \textit{maxLen}}
            \State add $S$ to \textit{LUSPs};
        \EndIf
    \EndFor
    \State \Return \textit{LUSPs}
    \end{algorithmic}
\end{algorithm}

\subsection{Pruning Strategies and Search Trees}

In LUSPM$_b$, direct utility calculation on a substantial amount of generated candidates incurs significant computational overhead. To address this problem, we propose two improved algorithms, LUSPM$_s$ and LUSPM$_e$, to greatly reduce computational overhead. In this section, we introduce the pruning strategies and search trees employed by the algorithms.

\subsubsection{Definitions and Pruning Strategies}

To efficiently mine low-utility sequential patterns, we introduce several pruning strategies used in LUSPM$_s$ and LUSPM$_e$, together with the corresponding definitions and proofs of theorems.

\begin{definition}[Sequence Shrinkage and Removed-index]
   \rm Sequence shrinkage generates subsequences by removing items from a super-sequence. For a sequence $S$ obtained by removing an item $i$ at position $j$, further shrinkage is applied only to items appearing after position $j$. The value $j$ - 1 is referred to as the removed-index of $S$. 
\end{definition}

For example, consider the sequence $\langle$$a$, $b$, $c$, $d$, $e$$\rangle$. By removing the item $b$ at position 2, the sequence $\langle$$a$, $c$, $d$, $e$$\rangle$ is obtained. By further removing items after position 2, the sequences $\langle$$a$, $d$, $e$$\rangle$, $\langle$$a$, $c$, $e$$\rangle$, and $\langle$$a$, $c$, $d$$\rangle$ can be obtained.

\begin{definition}[Sequence Extension \cite{ayres2002sequential}]\label{Sequence Extension}
  \rm Sequence extension generates super-sequences by inserting items after the last element of a subsequence. For a sequence $S$ ending with item $i$ and contained in a super-sequence $F$, any item appearing after $i$ in $F$ can be appended to generate a longer sequence. Starting from the empty set and extending iteratively produces all subsequences of $F$.
\end{definition}

For example, starting from $\langle$$a$, $b$, $c$, $d$, $e$$\rangle$, we begin with $\emptyset$, insert $a$ and $c$ to generate $\langle$$a$, $c$$\rangle$, and then extend it with $d$ to generate $\langle$$a$, $c$, $d$$\rangle$.

\begin{theorem}\label{Theorem 5}
    \rm For a sequence $F$ and an item $i$ at position $j$ in $F$, $u$($i$, $j$ - 1, $F$) $\leq$ $u$($F$).
\end{theorem}
\begin{proof}
    \rm Let $F$ = $\langle$$i_1$, $i_2$, $\ldots$, $i_n$$\rangle$ with sequence-utility chain $M$ = $\langle$$\langle$$a_{11}$, $a_{12}$, $\ldots$, $a_{1n}$$\rangle$, $\langle$$a_{21}$, $a_{22}$, $\ldots$, $a_{2n}$$\rangle$, $\ldots$, $\langle$$a_{k1}$, $a_{k2}$, $\ldots$, $a_{kn}$$\rangle$$\rangle$. Then $u$($a$, $j$ - 1, $F$) = $\sum_{p=1}^k a_{pj}$, while $u$($F$) = $\sum_{1 \leq q \leq n}$ $u$($i_{q}$, $q$ - 1, $F$). Hence, $u$($a$, $j$ - 1, $F$) $\leq$ $u$($F$).
\end{proof}

In Table \ref{table2}, the sequence $S$ = $\langle$$a$, $b$, $c$$\rangle$ has sequence-utility chain $\langle$$\langle$1, 2, 1$\rangle$, $\langle$1, 2, 2$\rangle$$\rangle$. Then $u$($b$, 1, $S$) = 2 + 2 = 4, and $u(S)$ = $u$($a$, 0, $S$) + $u$($b$, 1, $S$) + $u$($c$, 2, $S$) = 9 $>$ 4.

\begin{strategy}\label{strategy 1}
    \rm {Early Utility Pruning Strategy (EUPS):} For a sequence $F$ with item $i$ at position $j$, if $u$($i$, $j$ - 1, $F$) $>$ minUtil, then by Theorem~\ref{Theorem 5}, any low-utility sequence derived from $F$ cannot contain $i$. Thus, $i$ can be pruned. Removing $i$ yields a new sequence $Q$ to replace $F$. 
\end{strategy}
\begin{proof}
    \rm Let $F$ be a sequence and $i$ be an item at position $j$ in $F$. If $u$($i$, $j$ - 1, $F$) $>$ \textit{minUtil}, then for any sequence $Q$ derived from $F$ that contains this occurrence of $i$, whether by shrinking or extending $F$, we have $u(Q)$ $\geq$ $u$($i$, $j$ - 1, $F$) $>$ \textit{minUtil}. Consequently, $Q$ cannot be a LUSP, since its utility exceeds the threshold. Thus, it is valid to preemptively prune $i$ from $F$, resulting in a new sequence $Q'$ that replaces $F$.
\end{proof}

For example, with \textit{minUtil} = 3, the sequence $F$ = $\langle$$a$, $b$, $c$$\rangle$ has the sequence-utility chain $\langle$$\langle$1, 2, 1$\rangle$, $\langle$1, 2, 2$\rangle$$\rangle$ in Table~\ref{table2}. Since $u$($b$, 1, $F$) = 2 + 2 = 4 $>$ \textit{minUtil}, the item $b$ is pruned to obtain $Q$ = $\langle$$a$, $c$$\rangle$, which then replaces $F$.

\begin{definition}[Lower Bound within Super-sequence]
    \rm Let $F$ = $\langle$$i_1$, $i_2$, $\ldots$, $i_n$$\rangle$ be a sequence, $S$ $\subseteq$ $F$ be a subsequence generated by removing some items from $F$, and let $N$ denote the sequence-utility chain of $F$. By removing the utilities of the removed items in $N$, we obtain the sequence-utility chain $M$ = $\langle$$\langle$$a_{11}$, $a_{12}$, $\ldots$, $a_{1m}$$\rangle$, $\langle$$a_{21}$, $a_{22}$, $\ldots$, $a_{2m}$$\rangle$, $\ldots$, $\langle$$a_{k1}$, $a_{k2}$, $\ldots$, $a_{km}$$\rangle$$\rangle$ of $S$ within $F$. Based on $M$, we define the lower bound of $S$ within its super-sequence $F$ as 
\begin{flalign}
    &&
   LBS(S,F) = \sum_{(1 \leq i \leq sup(F)) \land (1 \leq j \leq m) } a_{ij},
    &&
\end{flalign}
where \textit{sup}$(F)$ is the number of occurrences of $F$ in the database, $m$ is the length of $S$, and $a_{ij}$ $\in$ $M$.
\end{definition}

For example, for $F$ = $\langle$$a$, $b$, $c$$\rangle$ with sequence-utility chain $N$ = $\langle$$\langle$1, 2, 1$\rangle$, $\langle$1, 2, 2$\rangle$$\rangle$, by removing items $a$ and $c$, we can generate $S$ = $\langle$$b$$\rangle$ with sequence-utility chain $M$ = $\langle $$\langle$2$\rangle$, $\langle$2$\rangle$$\rangle$. Then, \textit{LBS}($S$, $F$) = 2 + 2 = 4.

\begin{theorem}\rm\label{Theorem 1}
    For any sequence $F$ and its subsequence $S$, \textit{sup}($F$) $\leq$ \textit{sup}($S$).
\end{theorem}
\begin{proof} 
    \rm Since $S$ $\subseteq$ $F$, every occurrence of $F$ contains $S$. Hence, \textit{sup}($F$) $\leq$ \textit{sup}($S$). 
\end{proof}

For example, in Table \ref{table2}, the sequence $S$ = $\langle$$a$, $b$, $c$$\rangle$ has \textit{sup}$(S)$ = 2 and $Q$ = $\langle$$a$, $b$$\rangle$ has \textit{sup}$(Q)$ = 8. It always holds that \textit{sup}($Q$) $\geq$ \textit{sup}($S$) = 2.

\begin{theorem}\rm\label{Theorem 2}
    For any sequence $F$ and its subsequence $S$, it holds that \textit{LBS}($S$, $F$) $\leq$ $u$($S$).
\end{theorem}

\begin{proof}
    \rm  By Theorem \ref{Theorem 1}, we have \textit{sup}($F$) $\leq$ \textit{sup}($S$). If \textit{sup}($F$) = \textit{sup}($S$), then $S$ and $F$ co-occur in all cases, so \textit{LBS}($S$, $F$) = $u$($S$). If \textit{sup}($F$) $<$ \textit{sup}($S$), there exist $b$ occurrences where $S$ appears without $F$, implying $u(S)$ = \textit{LBS}($S$, $F$) + $\sum_{p=1}^{b}$ $\sum_{q=1}^{m}$$a_{pq}$, where $m$ is the number of items in $S$. Hence, \textit{LBS}($S$, $F$) $\leq$ $u$($S$).
\end{proof}

For example, in Table \ref{table2}, the sequence $S$ = $\langle$$a$, $b$, $c$$\rangle$ appears in $S_1$ and $S_2$. Its subsequence $Q$ = $\langle$$a$, $b$$\rangle$ appears eight times: three times in $S_1$, once in $S_2$, once in $S_3$, and three times in $S_4$. The sequence-utility chain of $Q$ within $S$ is $\langle $$\langle$1, 2$\rangle$, $\langle$1, 2$\rangle$$\rangle$. Thus, \textit{LBS}($Q$, $S$) = 1 + 2 + 1 + 2 = 6, while $u$($Q$) = 30 $>$ \textit{LBS}($Q$, $S$).

\begin{strategy}[Shrinkage-Based Low-Utility Sequence Pruning Strategy (SLUSPS]
\label{strategy 2}
    \rm When shrinking sequences, if $u(F)$ $>$ \textit{minUtil} and $S$ is generated from $F$ through shrinkage, we first compute \textit{LBS}($S$, $F$). If \textit{LBS}($S$, $F$) $>$ \textit{minUtil}, this implies that $S$ is not a LUSP and should be pruned. Otherwise, we check $u(S)$ to determine whether $S$ is LUSP, then further shrink $S$ to generate new subsequences. The same pruning procedure is recursively applied to each of them.
\end{strategy}
\begin{proof}
    \rm  Suppose $S$ is a subsequence generated by shrinking a super-sequence $F$. If \textit{LBS}($S$, $F$) $>$ \textit{minUtil}, then by Theorem~\ref{Theorem 2}, we have: $u$($S$) $\geq$ \textit{LBS}($S$, $F$) $>$ \textit{minUtil}. Hence, $S$ cannot be a LUSP, since its utility exceeds the threshold. Therefore, pruning $S$ at this stage is valid. If \textit{LBS}($S$, $F$) $\leq$ \textit{minUtil}, then \textit{LBS} alone cannot determine whether $S$ is a LUSP. In this case, we compute $u(S)$. If $u(S)$ $>$ \textit{minUtil}, $S$ is not a LUSP and can be pruned. Otherwise, $S$ is retained as a candidate LUSP, and the shrinking process continues recursively to generate further subsequences, to which the same pruning logic is applied.
\end{proof}

For example, let \textit{minUtil} = 3. In Table \ref{table2}, consider the sequence $F$ = $\langle$$d$, $a$, $c$$\rangle$ with $u$($F$) = 4 $>$ 3, indicating that $F$ is not a low-utility sequential pattern (LUSP). We then generate its subsequences $P$ = $\langle$$d$, $a$$\rangle$, $Q$ = $\langle$$d$, $c$$\rangle$, and $O$ = $\langle$$a$, $c$$\rangle$. For $P$, \textit{LBS}($P$, $F$) = 3, which cannot directly determine whether $P$ is a LUSP. Therefore, we compute its actual utility $u$($P$) = 3 =  \textit{minUtil}, indicating that $P$ is a LUSP. For $Q$, \textit{LBS}($Q$, $F$) = 5 $>$ \textit{minUtil}. According to Theorem~\ref{Theorem 2}, $Q$ cannot be a LUSP. Therefore, $Q$ is pruned without computing $u$($Q$). Shrinking $Q$ generates $B$ = $\langle$$d$$\rangle$ and $C$ = $\langle$$c$$\rangle$, where $u$($B$) = 11 $>$ \textit{minUtil} and $u$($C$) = 7 $>$ \textit{minUtil}. Hence, neither $B$ nor $C$ is a LUSP, and they are pruned. The sequence $O$ is processed in the same manner.

\begin{definition}[Determined Subsequence and Extension of Determined Subsequence]
    \label{Determined subsequence and extension of determined subsequence}
    \rm Let $F$ = $\langle$$i_1$, $i_2$, $\ldots$, $i_n$$\rangle$ and suppose that $S$ is generated by removing item $i_p$ from $F$. Then the prefix $P$ = $\langle$$i_{1}$, $i_{2}$, $\ldots$, $i_{p-1}$$\rangle$ is called the determined subsequence of $S$. Furthermore, any sequence generated by inserting an item from $S$ into $P$ after position $p$ - 1 is called an extension sequence of $P$, which is also a subsequence of $F$.
\end{definition}

For example, given $F$ = $\langle$$a$, $b$, $c$, $a$, $b$, $d$$\rangle$, removing $c$ yields the sequence $S$ = $\langle$$a$, $b$, $a$, $b$, $d$$\rangle$. The determined subsequence of $S$ is $P$ = $\langle$$a$, $b$$\rangle$. From $S$, extension sequences of $P$ such as $\langle$$a$, $b$, $a$$\rangle$, $\langle$$a$, $b$, $b$$\rangle$, and $\langle$$a$, $b$, $d$$\rangle$ can be derived, all of which are subsequences of $F$.

\begin{definition}[Lower Bound for Prune]\label{Lower Bound For Prune}
    \rm  Let $F$ = $\langle$$i_1$, $i_2$, $\ldots$, $i_n$$\rangle$ be a sequence, and let $S$ be a subsequence generated by removing item $i_p$ from $F$, with determined subsequence $P$. For any extension sequence $Q$ of $P$ generated by inserting $i_q$ ($q$ $\geq$ $p$) in $F$, the lower bound for prune of $S$ at position $q$ - 1 in $F$ is defined as \textit{LBP}($S$, $F$, $q$ - 1) = \textit{LBS}($Q$, $F$).
\end{definition}

For example, using the example from Definition \ref{Determined subsequence and extension of determined subsequence}, suppose that the sequence-utility chain of the sequence $F$ is $\langle$1, 2, 1, 2, 3, 3$\rangle$. Extending sequence $P$ by item $a$ at position 4 in $F$ yields \textit{LBP}($S$, $F$, 3) = 1 + 2 + 2 = 5.

\begin{theorem}\rm \label{Theorem 3}
    For any sequence $F$ and its subsequence $S$, we have \textit{LBS}($S$, $F$) $<$ $u(F)$.
\end{theorem}
\begin{proof}
    \rm Since the sequence-utility chain of $S$ with $F$ is contained within that of $F$, the sum of its elements must be strictly less than the total utility of $F$. Therefore, \textit{LBS}($S$, $F$) $<$ $u$($F$). 
\end{proof}

For example, in Table \ref{table2}, let $F$ = $\langle$$a$, $b$, $c$$\rangle$ with sequence-utility chain $\langle$$\langle$1, 2, 1$\rangle$, $\langle$1, 2, 2$\rangle$$\rangle$, where $u(F)$ = 9. For the subsequence $S$ = $\langle$$a$, $b$$\rangle$, its chain with $F$ is $ \langle$$\langle$1, 2$\rangle$, $\langle$1, 2$\rangle$$\rangle$, giving \textit{LBS}($S$, $F$) = 6 $<$ 9.

\begin{theorem}\label{Theorem 4}
    \rm \rm Let $F$ be a sequence and $S$ and $Q$ be subsequences of $F$ such that $Q$ is an extension of $S$. Then we have \textit{LBS}($S$, $F$) $<$ \textit{LBS}($Q$, $F$) $<$ $u(F)$.
\end{theorem}
\begin{proof}
    Since $S$ $\subset$ $Q$ $\subset$ $F$, the sequence-utility chain of $S$ within $F$ is contained in that of $Q$, which is in turn contained in that of $F$. Therefore, summing the corresponding elements of these chains yields the inequality.
\end{proof}

For example, in Table~\ref{table2}, let $F$ = $\langle$$a$, $b$, $c$$\rangle$ with sequence-utility chain $N$ = $\langle$$\langle$1, 2, 1$\rangle$, $\langle$1, 2, 2$\rangle$$\rangle$, $S$ = $\langle$$a$$\rangle$, and $Q$ = $\langle$$a$, $c$$\rangle$. We obtain \textit{LBS}($S$, $F$) = 2, \textit{LBS}($Q$, $F$) = 5, $u$($F$) = 9, which satisfies \textit{LBS}($S$, $F$) $<$ \textit{LBS}($Q$, $F$) $<$ $u$($F$).

\begin{strategy}[Shrinkage-Based Invalid Item Pruning (SBIPS)]
    \label{strategy 3}
    \rm For a sequence $S$ derived from $F$ with a determined subsequence $P$, if for an extension $Q$ = $P \oplus i_k$ we have \textit{LBP}($S$, $k$ - 1) = \textit{LBS}($Q$, $F$) $>$ \textit{minUtil}, then all sequences generated by further shrinking $S$ that contain item ${i}_{k}$ can be pruned. 
\end{strategy}
\begin{proof}
    \rm Let $S$ be a sequence generated by shrinking a super-sequence $F$, with an item $i_k$ such that \textit{LBP}($S$, $k$ - 1) $>$ \textit{minUtil}. By Definition~\ref{Lower Bound For Prune}, we have \textit{LBS}($Q$, $F$) = \textit{LBP}($S$, $k$ - 1) for the extension $Q$ = $P$ $\oplus$ $i_k$. Since $u(Q)$ $\geq$ \textit{LBS}($Q$, $F$) $>$ \textit{minUtil} by Theorem~\ref{Theorem 2}, $Q$ cannot be a LUSP. Furthermore, for any sequence $Q'$ generated by further shrinking $S$ that still contains item $j_k$, its utility satisfies $u(Q')$ $\geq$ $u$($Q$), because $Q'$ is a subsequence of $Q$ generated by removing other items while retaining $i_k$. Therefore, $u$($Q'$) $>$ \textit{minUtil} also holds. This means that $Q'$ cannot be a LUSP, and pruning item $i_k$ is valid.
\end{proof}

For example, in Definition~\ref{Determined subsequence and extension of determined subsequence}, let $\textit{minUtil} = 3$. If \textit{LBP}($S$, 2) = 5 $>$ 3 for the extension $Q$ = $\langle$$a$, $b$, $a$$\rangle$, then by Theorem~\ref{Theorem 2}, we have $u(Q)$ $\geq$ \textit{LBS}($Q$, $F$) = 5 $>$ \textit{minUtil}, indicating that $Q$ cannot be a LUSP. Moreover, any further subsequence of $Q$ that still contains item $a$ must also have utility exceeding \textit{minUtil}, so $a$ can be safely pruned from $S$.

\begin{strategy}[Expansion-Based Invalid Sequence Pruning Strategy (EBISPS)]
    \label{strategy 4}
   \rm  For a sequence $F$, if there exists a subsequence $S$ such that \textit{LBS}($S$, $F$) $>$ \textit{minUtil}, then $S$, $F$, and any sequence $Q$ generated by extending $F$ can be pruned. By Theorem~\ref{Theorem 2}, we have $u(S)$ $\geq$ \textit{LBS}($S$, $F$) $>$ \textit{minUtil}, indicating that $S$ cannot be a LUSP. Furthermore, Theorem~\ref{Theorem 3} implies $u(F)$ $>$ \textit{minUtil}, and Theorem~\ref{Theorem 4} guarantees that for any extension $Q$ of $F$, $u(Q)$ $\geq$ \textit{LBS}($Q$, $F$) $>$ \textit{minUtil}. Therefore, pruning $S$, $F$, and all their extensions is valid.
\end{strategy}
\begin{proof}
    \rm Suppose $F$ is a super-sequence and $S$ is a subsequence of $F$. If \textit{LBS}($S$, $F$) $>$ \textit{minUtil}, then by Theorem~\ref{Theorem 2}, we have $u(S)$ $\geq$ \textit{LBS}($S$, $F$) $>$ \textit{minUtil}, which means $S$ is not a LUSP. Moreover, by Theorem~\ref{Theorem 3}, the utility of $F$ also exceeds \textit{minUtil}, so $F$ is not a LUSP. Finally, Theorem~\ref{Theorem 4} ensures that for any extension $Q$ of $F$, $u(Q)$ $\geq$ \textit{LBS}($Q$, $F$) $>$ \textit{minUtil}. Thus, $Q$ cannot be a LUSP either. Consequently, pruning $S$, $F$, and all extensions $Q$ derived from $F$ is justified.
\end{proof}

For example, let \textit{minUtil} = 3 and consider $F$ = $\langle$$a$, $b$, $c$, $a$, $b$$\rangle$ with sequence-utility chain $N$ = $\langle$1, 2, 1, 2, 3$\rangle$. For the subsequence $S$ = $\langle$$a$, $b$, $c$$\rangle$, we calculate \textit{LBS}($S$, $F$) = 4 $>$ \textit{minUtil}, indicating that $S$ is not a LUSP. Next, for the extension $Q$ = $\langle$$a$, $b$, $c$, $a$$\rangle$, we have \textit{LBS}($Q$, $F$) = 6 $>$ \textit{minUtil}, so $Q$ is not a LUSP either. Further extending to $F$ yields $u(F)$ = 9 $>$ \textit{minUtil}. Therefore, $S$, $F$, and all extensions derived from $F$ can be safely pruned.

\begin{figure*}[htbp]
    \centering
    \includegraphics[width=0.90\textwidth]{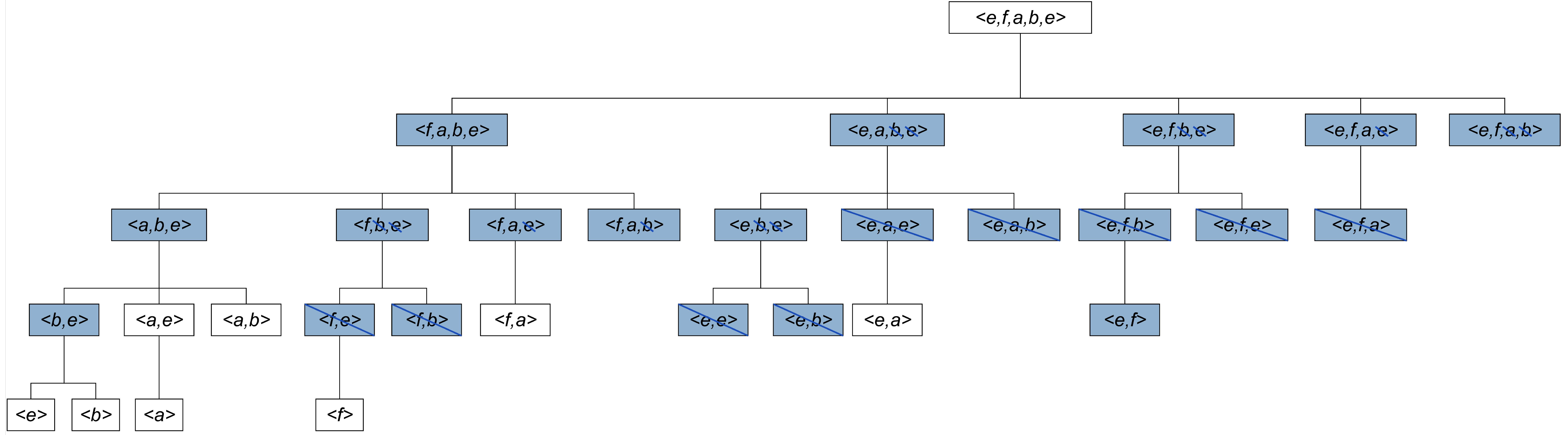} 
    \caption{Shrinkage search tree of sequence $\langle$$e$, $f$, $a$, $b$, $e$$\rangle$ when \textit{minUtil} = 3.}
    \label{fig:Shrinkage Search Tree}
\end{figure*}

\begin{figure*}[htbp]
    \centering
    \includegraphics[width=0.90\textwidth ]{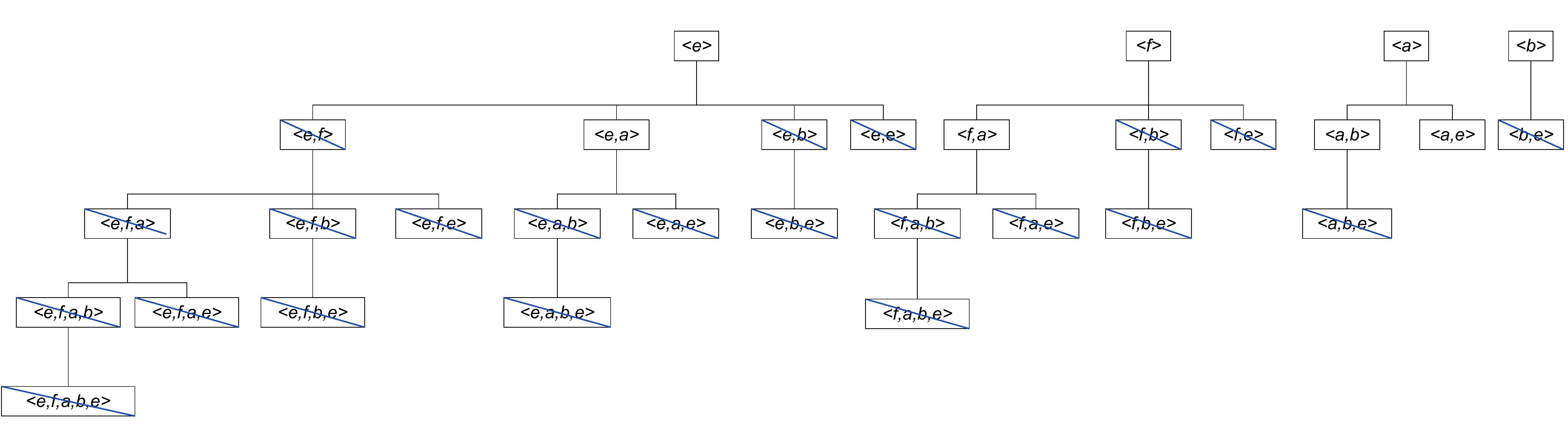}
    \caption{Extension search tree of sequence $\langle$$e$, $f$, $a$, $b$, $e$$\rangle$ when \textit{minUtil} = 3.}
    \label{fig:Extension Search Tree}
\end{figure*}

\subsubsection{Search Trees} \label{search tree}

To effectively explore the search space of low-utility candidate sequences, we propose two search tree structures: the shrinkage search tree and the extension search tree, corresponding to LUSPM$_s$ and LUSPM$_e$, respectively. In the shrinkage search tree, candidate sequences are generated by removing items from their super-sequences. Starting from the original sequence as the root, each child node is obtained by removing a single item, and the tree is recursively expanded to form a layer-wise shrinkage-based subsequence space. Fig.~\ref{fig:Shrinkage Search Tree} illustrates this construction using the sequence $\langle$$e$, $f$, $a$, $b$, $e$$\rangle$ as an example (based on Table~\ref{table2} with \textit{minUtil} = 3). To improve efficiency, LUSPM$_s$ incorporates pruning strategies \ref{strategy 2} and \ref{strategy 3}, where nodes with colored backgrounds are pruned by strategy \ref{strategy 2}, and items and sequences marked with slashes are pruned by strategy \ref{strategy 3}. In the extension search tree, candidate sequences are generated by incrementally inserting items into subsequences. Starting from the empty sequence as the root, each child node is obtained by inserting a single item, and the tree is recursively expanded to form a layer-wise extension-based sequence space. Fig.~\ref{fig:Extension Search Tree} illustrates the construction under the same settings. LUSPM$_e$ employs the pruning strategy \ref{strategy 4}, where sequences with strikethroughs (i.e., those whose utilities exceed \textit{minUtil}) are identified as invalid and pruned.

\subsection{Algorithm Details}

We present two improved algorithms, LUSPM$_s$ and LUSPM$_e$, for the efficient mining of LUSPs. We begin by describing the preprocessing steps shared by both algorithms, and then detail the procedures of each algorithm individually.

\subsubsection{Prune By Preprocessing}\label{prune by preprocessing}The complexity of the search forest in the algorithm is related to the number of sequences in the database, where each sequence corresponds to a search tree. As shown in Fig. \ref{fig:Shrinkage Search Tree}, a sequence of length $m$ can generate up to $m!$ subsequences. However, inclusion relationships may exist between search trees. For example, in Table~\ref{table2}, the tree of $S_6$ is a subtree of the tree of $S_1$. Inspired by the maximal non-mutually contained itemset in the LUIM algorithm \cite{zhang2025enabling}, we propose the concept of maximal non-mutually contained sequence to improve mining efficiency.

\begin{definition}
    \rm For a sequence database $\mathcal{D}$, a subset $M \subseteq \mathcal{D}$ is called the \textit{Maximal Non-Mutually Contained Sequence Set} (abbreviated as \textit{MaxNonConSeqSet}) of $\mathcal{D}$, if every sequence in $\mathcal{D}$ is a subsequence of some sequence in $M$, and no sequence in $M$ is a subsequence of another sequence in $M$. Each sequence in $M$ is referred to as a \textit{Maximal Non-Mutually Contained Sequence} (abbreviated as \textit{MaxNonConSeq}) of $\mathcal{D}$.
\end{definition}

In Table \ref{table2}, $S_1$, $S_2$, $S_3$, $S_4$, and $S_5$ do not mutually contain each other, whereas $S_6$ is a subsequence of $S_1$. Consequently $S_1$, $S_2$, $S_3$, $S_4$, and $S_5$ are \textit{MaxNonConSeq} and the \textit{MaxNonConSeqSet} is $M$ = \{$S_1$, $S_2$, $S_3$, $S_4$, $S_5$\}. Based on Strategy \ref{strategy 1} and the concept of \textit{MaxNonConSeqSet}, we propose Algorithm \ref{preprocess} as a preprocessing step in both LUSPM$_s$ and LUSPM$_e$. This algorithm first prunes items in the sequences of the database using Strategy \ref{strategy 1}, and then applies a deduplication step to obtain the final \textit{MaxNonConSeqSet}. The algorithm requires two inputs: the sequence database $\mathcal{D}$ and the minimum utility threshold \textit{minUtil}. For each sequence \textit{S} in $\mathcal{D}$, its sequence-utility chain \textit{utilChain} is obtained. For each item in \textit{S}, the corresponding utility sum is calculated from \textit{utilChain}. If this sum exceeds \textit{minUtil}, the item is considered invalid according to Strategy \ref{strategy 1} and is therefore removed. The remaining sequence \textit{S} is stored in \textit{maxNonConSeqSet} (lines 1–11). Finally, all sequences in \textit{maxNonConSeqSet} are checked, and any sequence that is a subsequence of another is removed to ensure that every sequence is a \textit{MaxNonConSeq} (lines 12–16).

\begin{algorithm}[!h]
    \small
    \caption{preprocess}
    \label{preprocess}
    \begin{algorithmic}[1]
    \Statex \textbf{Input:} {$\mathcal{D}$: a sequence database; \textit{minUtil}: utility threshold}
    \Statex \textbf{Output:} {$\mathit{maxNonConSeqSet}$: set of MaxNonConSeq}
    \For{\rm each sequence $S \in \mathcal{D}$}
        \State $\mathit{utilChain}$ = $\mathit{getUtilityChain}(S)$;
        \For{\rm $i$ = 0 $\textbf{to}$ $|S| - 1$}
            \State $\mathit{itemUtil}$ = $\sum_{u \in \mathit{utilChain}} u[i]$;
            \If{\rm $\mathit{itemUtil} > \mathit{minUtil}$}
                \State remove $i$th item from $S$;
               \State remove $i$th element from each $u \in \mathit{utilChain}$;
                \State $i$ = $i$ - 1;
            \EndIf
        \EndFor
        \State Add $S$ to $\mathit{maxNonConSeqSet}$;
    \EndFor
    \For{\rm $S$, $Q \in \mathit{maxNonConSeqSet}$}
        \If{\rm $S \neq Q \land Q \preceq S$}\hfill \textit{// $Q$ is the subsequence of $S$}
            \State $\mathit{maxNonConSeqSet} = \mathit{maxNonConSeqSet} \setminus \{S\}$;
        \EndIf
    \EndFor
    \State \Return $\mathit{maxNonConSeqSet}$
    \end{algorithmic}
\end{algorithm}

\subsubsection{The LUSPM$_s$ Algorithm}

To discover all LUSPs more efficiently, we propose the LUSPM$_{s}$ algorithm. It leverages Strategies \ref{strategy 2} and \ref{strategy 3} to generate shorter sequences from longer ones. The pseudocode is provided in Algorithm \ref{The LUSPM$_s$ algorithm}. LUSPM$_{s}$ employs several functions: \textit{getUtilityChain}, which obtains the utility chain of a sequence; \textit{computeUtility}, which calculates the utility of a sequence; \textit{shrinkage} (Algorithm \ref{shrinkage}), which generates shorter sequences from longer ones and finds LUSPs; \textit{shrinkage$_{depth}$} (Algorithm \ref{shrinkage_Depth}), which reduces unnecessary utility computations based on Strategy \ref{strategy 2} during shrinkage; and \textit{pruneItem} (Algorithm \ref{pruneItem}), which removes invalid items from sequences using Strategy \ref{strategy 3}.

Algorithm \ref{The LUSPM$_s$ algorithm} describes the complete process of mining LUSPs through shrinkage. It takes a sequence database $\mathcal{D}$, \textit{minUtil}, and \textit{maxLen} as inputs, and outputs all LUSPs. First, the algorithm obtains the \textit{maxNonConSeqSet} of $\mathcal{D}$ (line 1). For each sequence \textit{S} in this set, it retrieves \textit{S}'s sequence-utility chain and calculates its utility. If the utility of sequence \textit{S} is not greater than \textit{minUtil}, the \textit{shrinkage} function is called to generate subsequences of \textit{S} by removing items, thereby obtaining additional LUSPs. Moreover, if the length of \textit{S} is not greater than \textit{maxLen}, \textit{S} is also stored as a LUSP (lines 2–9). Otherwise, if the utility of \textit{S} exceeds \textit{minUtil}, \textit{shrinkage$_{depth}$} is invoked according to Strategy \ref{strategy 2}, which generates subsequences of \textit{S} by removing items and leverages the partial utility of \textit{S} to reduce unnecessary utility computations, thereby obtaining more LUSPs (line 10).

\begin{algorithm}[!h]
    \small
    \caption{LUSPM$_s$}
    \label{The LUSPM$_s$ algorithm}
    \begin{algorithmic}[1]
    \Statex \textbf{Input:} $\mathcal{D}$: a sequence database; \textit{minUtil}: utility threshold; \textit{maxLen}: length restriction.
    \Statex \textbf{Output:} \textit{LUSPs}: the complete set of LUSP.
    
    \State initialize \textit{LUSPs} = $\emptyset$, \textit{maxNonConSeqSet} = \textit{preprocess}($\mathcal{D}$)
    \For{each sequence $\textit{S}$ $\in$ \textit{maxNonConSeqSet}}
        \State \textit{utilChain} = \textit{getUtilityChain}($\textit{S}$)
        \If{\textit{computeUtility}(\textit{utilChain}) $\leq$ \textit{minUtil}}
        \State call \textit{shrinkage}($\textit{S}$, 0, \textit{LUSPs})\hfill \textit{// Algorithm \ref{shrinkage}}
        \If{$|\textit{S}|$ $\leq$ $\textit{maxLen}$}
        \State add $\textit{S}$ to \textit{LUSPs}
        \EndIf    
        \EndIf
        \State call \textit{shrinkage$_{depth}$}($\textit{S}$, $\textit{utilityChain}$, 0, \textit{LUSPs})\hfill \textit{// Algorithm \ref{shrinkage_Depth}}
    \EndFor
    \State \Return \textit{LUSPs}
    \end{algorithmic}
\end{algorithm}

Algorithm \ref{shrinkage} describes the process of generating subsequences and mining LUSPs by removing items from longer sequences. It takes three inputs: sequence \textit{S}, its removed index \textit{p}, and LUSPs. First, it removes the \textit{p}-th item from \textit{S} to generate a new sequence \textit{Q} and its sequence-utility chain. After calculating the utility of \textit{Q}, it determines whether \textit{Q} is a LUSP. If the utility of \textit{Q} isn't greater than \textit{minUtil}, \textit{shrinkage} is called to generate subsequences of \textit{Q}. If the length of \textit{Q} also satisfies the length constraint, \textit{Q} is stored as a LUSP (lines 2-10). Otherwise, if the utility of \textit{Q} is greater than \textit{minUtil}, \textit{shrinkage$_{depth}$} is invoked according to Strategy \ref{strategy 2} (lines 12–14). Finally, if \textit{p} does not point to the last item in \textit{S}, it recursively calls \textit{shrinkage} to obtain additional candidate subsequences (lines 17-19).

\begin{algorithm}[!h]
    \small
    \caption{shrinkage}
    \label{shrinkage}
    \begin{algorithmic}[1] 
    \Statex \textbf{Input:} $S$: sequence; \textit{p}: removed-index of $S$; \textit{LUSPs}: the complete set of LUSP.
    \If{\textit{p} $<$ $|S|$}
        \State $Q$ = $S$; remove the \textit{p}th item in $Q$;
        \State \textit{utilChain} = \textit{getUtilityChain}($Q$);
        \If{\textit{computeUtility}(\textit{utilChain}, $|Q|$) $\leq$ \textit{minUtil}}
                    \If{\textit{p} $<$ $|Q|$}
                \State call \textit{shrinkage}($Q$, \textit{p}, \textit{LUSPs});
            \EndIf
            \If{$|Q|$ $\leq$ \textit{maxLen}}
                \State add $Q$ to \textit{LUSPs};
            \EndIf
        \Else
            \If{\textit{p} $<$ $|Q|$} \hfill \textit{ // Algorithm \ref{shrinkage_Depth}}
                \State call \textit{shrinkage$_{depth}$}($Q$, \textit{utilityChain}, \textit{p}, \textit{LUSPs});
            \EndIf
        \EndIf
    \EndIf
    \If{\textit{p} + 1 $<$ $|S|$}
        \State call \textit{shrinkage}($S$, \textit{p} + 1, \textit{LUSPs});
    \EndIf
    \end{algorithmic}
\end{algorithm}

Algorithm \ref{shrinkage_Depth} describes the process of generating subsequences and mining LUSPs using Strategy \ref{strategy 2}. The algorithm takes four inputs: a sequence \textit{S}, its sequence-utility chain \textit{utilChain}, a removed index \textit{p}, and LUSPs. First, if \textit{p} is within bounds, it calls the \textit{pruneItem} method to remove invalid items (lines 1–3). Next, it removes the \textit{p}-th item from both \textit{S} and \textit{utilChain}, producing a new sequence \textit{Q} and a new sequence-utility chain \textit{newChain} (lines 5–7). If \textit{Q} satisfies the length constraint, the utility of \textit{newChain} is evaluated. When this utility is not greater than \textit{minUtil}, the true utility of \textit{Q} is computed to determine whether \textit{Q} is a LUSP. If the true utility also does not exceed \textit{minUtil}, \textit{Q} is stored as a LUSP, and \textit{shrinkage} is called to discover its subsequences; otherwise, \textit{shrinkage$_{depth}$} is invoked under Strategy \ref{strategy 2} (lines 8–18). If the utility of \textit{newChain} exceeds \textit{minUtil}, \textit{shrinkage$_{depth}$} is again applied to process subsequences of \textit{Q} (lines 20–22). Then, if \textit{Q} fails to meet the length constraint, \textit{shrinkage$_{depth}$} is still executed to generate its subsequences (lines 25–27). Finally, if \textit{p} does not point to the last item in \textit{S}, it recursively calls \textit{shrinkage$_{depth}$} to generate subsequences (lines 30–32).

\begin{algorithm}[!h]
    \small
    \caption{shrinkage$_{depth}$}
    \label{shrinkage_Depth}
    \begin{algorithmic}[1] 
    \Statex \textbf{Input:} $S$: sequence; \textit{utilChain}: sequence-utility chain of $S$; \textit{p}: removed-index of $S$; \textit{LUSPs}: the complete set of LUSP.
    \If{\textit{p} $<$ $|S|$}
        \State call \textit{pruneItem}($S$, \textit{utilChain}, \textit{p});\hfill \textit{// Algorithm \ref{pruneItem}}
    \EndIf
    \If{\textit{p} $<$ $|S|$}
        \State $Q$ = $S$; remove the \textit{p}th item of $Q$;
        \State \textit{newChain} = \textit{utilityChain};
        \State remove the \textit{p}th item of \textit{utilities} $\in$ \textit{newChain};
        \If{$|Q| \leq \textit{maxLen}$}
            \If{\textit{computerUtility}(\textit{newChain}, $|Q|$) $\leq$ \textit{minUtil}}
                \State \textit{chain} = \textit{getUtilityChain}($Q$);
                \If{\textit{computerUtility}(\textit{chain}, $|Q|$) $\leq$ \textit{minUtil}}
                    \State add $Q$ to \textit{LUSPs};
                    \State call \textit{shrinkage}($Q$, \textit{p}, \textit{LUSPs});
                \Else
                    \If{\textit{p} $<$ $|S|$}
                        \State call \textit{shrinkage$_{depth}$}($Q$, \textit{chain}, \textit{p}, \textit{LUSPs});
                    \EndIf
                \EndIf
            \Else
                \If{\textit{p} $<$ $|S|$}
                    \State call \textit{shrinkage$_{depth}$}($Q$, \textit{chain}, \textit{p}, \textit{LUSPs});
                \EndIf
            \EndIf
        \Else
            \If{\textit{p} $<$ $|S|$}
                \State call \textit{shrinkage$_{depth}$}($Q$, \textit{chain}, \textit{p}, \textit{LUSPs});
            \EndIf
        \EndIf
    \EndIf
     \If{\textit{p} + 1 $<$ $|S|$}
        \State call \textit{shrinkage$_{depth}$}($S$, \textit{utilChain}, \textit{p} + 1, \textit{LUSPs});
    \EndIf
    \end{algorithmic}
\end{algorithm}

Algorithm \ref{pruneItem} describes the process of pruning invalid items using pruning Strategy \ref{strategy 3}. The algorithm takes three inputs: a sequence \textit{S}, its sequence-utility chain \textit{utilChain} (or that of its super-sequence), and a removed index \textit{p}. First, it initializes \textit{removedId} and \textit{utility} (line 1). Next, for each index \textit{i} from \textit{p} to the last position in \textit{S}, the algorithm determines the initial value of \textit{utility}: when \textit{p} is 0, \textit{utility} is set to 0 (lines 3–5); otherwise, \textit{utility} is computed as the sum of the first \textit{p} entries in \textit{utilChain} (lines 6–8). Here, the value of \textit{utility} equals the sum of the utilities of the first \textit{p} items in the sequence. Then, the \textit{i}-th utility from \textit{utilChain} is added to \textit{utility} (line 9). At this step, the value of \textit{utility} equals the sum of the utilities of the first \textit{p} items and the \textit{i}-th item in the sequence. Finally, if \textit{utility} exceeds \textit{minUtil}, the corresponding item is pruned as invalid according to Strategy \ref{strategy 3} (lines 10–16).

\begin{algorithm}[!h]
    \small
    \caption{pruneItem}
    \label{pruneItem}
    \begin{algorithmic}[1]
    \Statex \textbf{Input:} $S$: sequence; \textit{utilChain}: a sequence-utility chain of $S$ (or of a super-sequence of $S$); \textit{p}: removed-index of $S$.
    \State initialize \textit{removedId} = $\emptyset$, \textit{utility} = 0;
    \For{\textit{i} = \textit{p} \textbf{to} $|S|$ - 1}
        \If{\textit{p} == 0}
            \State \textit{utility} = 0;
        \EndIf
        \If{\textit{p} $>$ 0}
            \State \textit{utility} = \textit{computeUtility}(\textit{utilChain}, \textit{p});
        \EndIf
        \State \textit{utility} = \textit{utility} + \textit{sum}(\textit{u[i]} \textbf{for} $u$ \textbf{in} \textit{utilChain});
        \If{\textit{utility} $>$ \textit{minUtil}}
            \State add \textit{i} to \textit{removedId};
        \EndIf
    \EndFor
    \For{each \textit{j} $\in$ \textit{removedId}}
        \State remove \textit{j}-th item $\in$ $S$;
        \State remove \textit{j}-th element of each \textit{utility} $\in$ \textit{utilChain};
    \EndFor
    \end{algorithmic}
\end{algorithm}

\subsubsection{The LUSPM$_{e}$ Algorithm} Unlike the LUSPM$_s$ algorithm, which generates shorter sequences by removing items from longer sequences using Strategies \ref{strategy 2} and \ref{strategy 3}, LUSPM$_e$ generates longer sequences by inserting items into shorter ones and employs Strategy \ref{strategy 4} to effectively prune a large number of invalid sequences. Algorithm \ref{The LUSPM$_{e}$ algorithm} presents the complete process of mining LUSPs through extension. It takes a sequence database $\mathcal{D}$, \textit{minUtil}, and \textit{maxLen} as inputs, and outputs all LUSPs. Specifically, it first scans $\mathcal{D}$ to obtain the MaxNonConSeqSet. For each sequence \textit{S} in this set, the algorithm retrieves its sequence-utility chain and executes an extension function (i.e., Algorithm \ref{Extension}), starting from an empty set to generate longer sequences, thereby obtaining the complete set of LUSPs.

\begin{algorithm}[!h]
    \small
    \caption{LUSPM$_e$}
    \label{The LUSPM$_{e}$ algorithm}
    \begin{algorithmic}[1] 
    \Statex \textbf{Input:} $\mathcal{D}$: a sequence database; \textit{minUtil}: utility threshold; \textit{maxLen}: length restriction.
    \Statex \textbf{Output:} \textit{LUSPs}: the complete set of \textit{LUSP}.
    \State initialize \textit{LUSPs} = $\emptyset$, \textit{maxNonConSeqSet} = \textit{preprocess}($\mathcal{D}$);
    \For{each sequence $\textit{S}$ $\in$ \textit{maxNonConSeqSet}}
        \State \textit{utilChain} = \textit{getUtilityChain}($\textit{S}$);
        \State call \textit{extension}($\textit{S}, \textit{utilChain}, \emptyset$);\hfill \textit{// Algorithm \ref{Extension}}
    \EndFor
    \State \Return \textit{LUSPs}
    \end{algorithmic}
\end{algorithm}

\begin{algorithm}[!h]
    \small
    \caption{extension}
    \label{Extension}
    \begin{algorithmic}[1] 
    \Statex \textbf{Input:} $S$: sequence; \textit{utilChain}: sequence-utility chain of $S$; $Q$: subsequence of $S$.
    \State \textit{p} = $|Q|$; \textit{S'} = S;
    \If{\textit{p} + 1 $<$ $|S|$}
        \State remove the \textit{p}th item of $S'$;
        \State \textit{newChain} = \textit{utilChain};
        \State remove the \textit{p}th element of \textit{utilities} $\in$ \textit{newChain};
        \State call \textit{extension}($S'$, \textit{newChain}, $Q$);
        \State insert the \textit{p}th item to $Q$;
        \If{\textit{computerUtility}(\textit{utilChain}, $|Q|$) $\leq$ \textit{minUtil}}
        \State call \textit{extension}($S$, \textit{utilChain}, $Q$);
            \State \textit{newChain'} = \textit{getUtilityChain}($Q$);
            \If{$|Q|$ $\leq$ \textit{maxLen}}
                \If{\textit{computerUtility}(\textit{newChain'}, $|Q|$) $\leq$ \textit{minUtil}}
                    \State add $Q$ to \textit{LUSPs};
                \EndIf
            \EndIf
        \EndIf
    \EndIf
    \end{algorithmic}
\end{algorithm}

Algorithm \ref{Extension} is the process of generating longer sequences from shorter ones and mining LUSPs by using Strategy \ref{strategy 4} to prune invalid sequences. It takes three inputs: a sequence \textit{S}, its corresponding \textit{utilChain}, and a subsequence \textit{Q}. First, the algorithm initializes \textit{p}. If \textit{p} does not point to the last item of \textit{S}, it removes the \textit{p}-th item from \textit{S} and its utility from \textit{utilChain}, and recursively calls the extension function to generate subsequences from \textit{S} (lines 1–4). Next, the algorithm inserts the \textit{p}-th item of \textit{S} into \textit{Q}. If \textit{Q}’s utility in \textit{utilChain} is not greater than \textit{minUtil}, it recursively calls the extension function to generate new sequences. If \textit{Q} satisfies the length constraint, the algorithm computes its true utility to determine whether \textit{Q} is a LUSP. If the true utility of \textit{Q} does not exceed \textit{minUtil}, it stores \textit{Q} as a LUSP (lines 6–13). Otherwise, if the true utility of \textit{Q} exceeds \textit{minUtil}, Strategy \ref{strategy 4} indicates that all sequences extended from \textit{Q} are invalid, so there is no need to call the \textit{extension} method to generate further sequences.

\subsection{Complexity Analysis} \label{Complexity Anaysis}

We finally analyze the time and space complexity of LUSPM$_s$ and LUSPM$_e$. In the database loading phase, the sequence database is scanned once to construct utility mapping tables, with a time complexity of $O$($N$ $\times$ $M$), where $N$ is the number of sequences and $M$ is the average sequence length. In the maximal sequence generation phase, deduplication and item-level pruning are performed. Strategy~\ref{strategy 1} is applied to remove items whose utility exceeds $minUtil$, reducing the effective values of $S$ and $L$. To identify the maximal non-mutually contained sequence set, pairwise sequence comparisons are required to determine containment relationships among sequences. In the worst case, this process requires $O$($S^2$) comparisons, and each comparison takes $O$($L$) time to check sequence containment, resulting in an overall time complexity of $O$($S^2$ $\times$ $L$), where $S$ denotes the number of unique sequences after deduplication and $L$ represents the average sequence length. 

In the core mining phase, the two algorithms adopt different recursive search strategies. LUSPM$_s$ enumerates candidate subsequences through the \textit{shrinkage} and $\textit{shrinkage}_{depth}$ operations and is further optimized by Strategy~\ref{strategy 2} and Strategy~\ref{strategy 3}. In contrast, LUSPM$_e$ adopts an extension-based enumeration strategy and is optimized by Strategy~\ref{strategy 4}. For a sequence of length $K$, the subsequence search space may reach $O$($2^K$) in the worst case. Thus, the worst-case time complexity of the mining phase is $O$($S$ $\times$ $2^K$ $\times$ $U$), where $U$ denotes the cost of a utility computation. Bitmap indexing is employed to accelerate sequence matching during utility computation. Overall, the worst-case time complexity of both algorithms is: $O$($N$ $\times$ $M$ + $S^2$ $\times$ $L$ + $S$ $\times$ $2^K$ $\times$ $U$). However, due to the pruning strategies, the practical search space is greatly reduced, and the average-case complexity can be approximated as $O$($N'$ $\times$ $M'$ + $S'$ $\times$ $K$ $\times$ $F$), where $N'$ $\leq$ $N$, $M'$ $\leq$ $M$, and $S'$ $<$ $S$ denote the reduced dataset and candidate space after pruning, and $F$ represents the average item frequency.

Regarding space complexity, storing the sequence database requires $O$($N$ $\times$ $M$) space, while the bitmap index requires $O$($I$ $\times$ $B$) space, where $I$ is the number of distinct items and $B$ is the bitmap length. The discovered pattern set requires $O$($P$ $\times$ $L$) space, and the recursion stack requires at most $O$($K$) space. In practice, early filtering further reduces transaction storage to $O$($N'$ $\times$ $M'$). The exponential worst-case complexity is inherent to sequential pattern mining due to the combinatorial nature of subsequence enumeration.

\section{Experimental Results and Analysis}
\label{sec: experiments}

In this section, we present the experimental evaluation of the proposed LUSPM$_{b}$, LUSPM$_{s}$, and LUSPM$_{e}$ across various datasets. We first describe the datasets and then compare LUSPM$_s$ and LUSPM$_e$ with LUSPM$_b$ in terms of runtime, memory usage, utility computation, and scalability across different settings, including varying \textit{minUtil} thresholds and sequence-length constraints. To ensure fairness, we compare the performance of the algorithms under the condition that all of them produce consistent mining results. All experiments were conducted on a Windows 10 PC with an Intel i7-10700F CPU and 16GB of RAM. The source code and datasets are available at https://github.com/Zhidong-Lin/LUSPM.

\subsection{Datasets Description}
We evaluate the proposed algorithms on several publicly available datasets, including four real-world datasets (SIGN, Leviathan, Kosarak10k, and Bible) and two synthetic datasets (Synthetic3k and Synthetic8k). These datasets span diverse scenarios, thereby enabling a comprehensive evaluation of our methods. All datasets are obtained from the SPMF repository\footnote{\url{https://www.philippe-fournier-viger.com/spmf}}. Table \ref{Tab:data} summarizes their characteristics, including the number of sequences and items, the maximum and average sequence lengths, and the total utility. For clarity, the datasets are listed in ascending order based on the number of sequences.

\vspace{-10pt}
\begin{table}[h!]
    \centering
    \scriptsize 
    \caption{The characteristics of the datasets}
    \label{Tab:data}
    \begin{tabular}{|c|>{\centering\arraybackslash}p{10mm}|>{\centering\arraybackslash}p{6mm}|>{\centering\arraybackslash}p{8mm}|>{\centering\arraybackslash}p{8mm}|>{\centering\arraybackslash}p{11mm}|}
        \hline
        \textbf{Dataset} & \textbf{Sequences} & \textbf{Items} & \textbf{MaxLen}& \textbf{AvgLen} & \textbf{TotalUtility} \\
          \hline
        SIGN & 730 & 267 & 94 & 51.997 & 634,332 \\
          \hline
        Synthetic3k & 3,196 & 75 & 36 & 36.000 & 2,156,659 \\
           \hline
        Leviathan & 5,834 & 9,025 & 72 & 33.810 & 1,199,198 \\
        \hline
        Synthetic8k & 8,124 & 119 & 22 & 22.0 & 3,413,720 \\
        \hline
        Kosarak10k & 10,000 & 10,094 & 608 & 8.140 & 1,396,290 \\
         \hline
        Bible & 36,369 & 13,905 & 77 & 21.641 & 12,817,639 \\
        \hline
    \end{tabular}
\end{table}

\subsection{Efficiency Analysis}
We first compare the efficiency of LUSPM$_b$, LUSPM$_s$, and LUSPM$_e$ under varying \textit{minUtil} values without a maximum length constraint.

\begin{figure*}[t!]
    \centering
\includegraphics[width=0.85\textwidth]{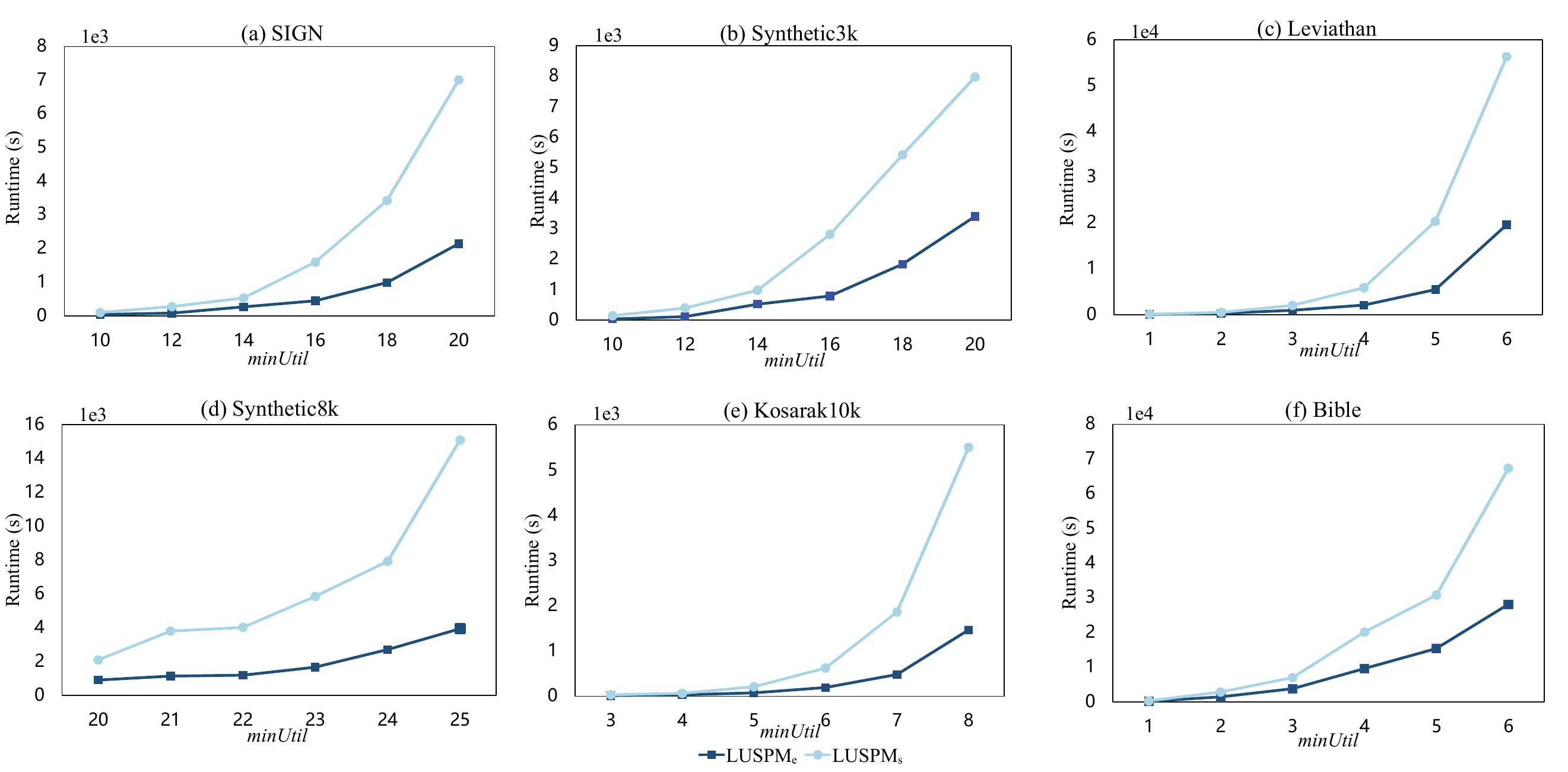}
    \caption{Time consumption analysis of LUSPM$_{s}$ and LUSPM$_{e}$.}
    \label{fig:LUSPM Time Consumption Analysis}
\end{figure*}
\begin{figure*}[t!]
    \centering
    \includegraphics[width=0.85\textwidth]{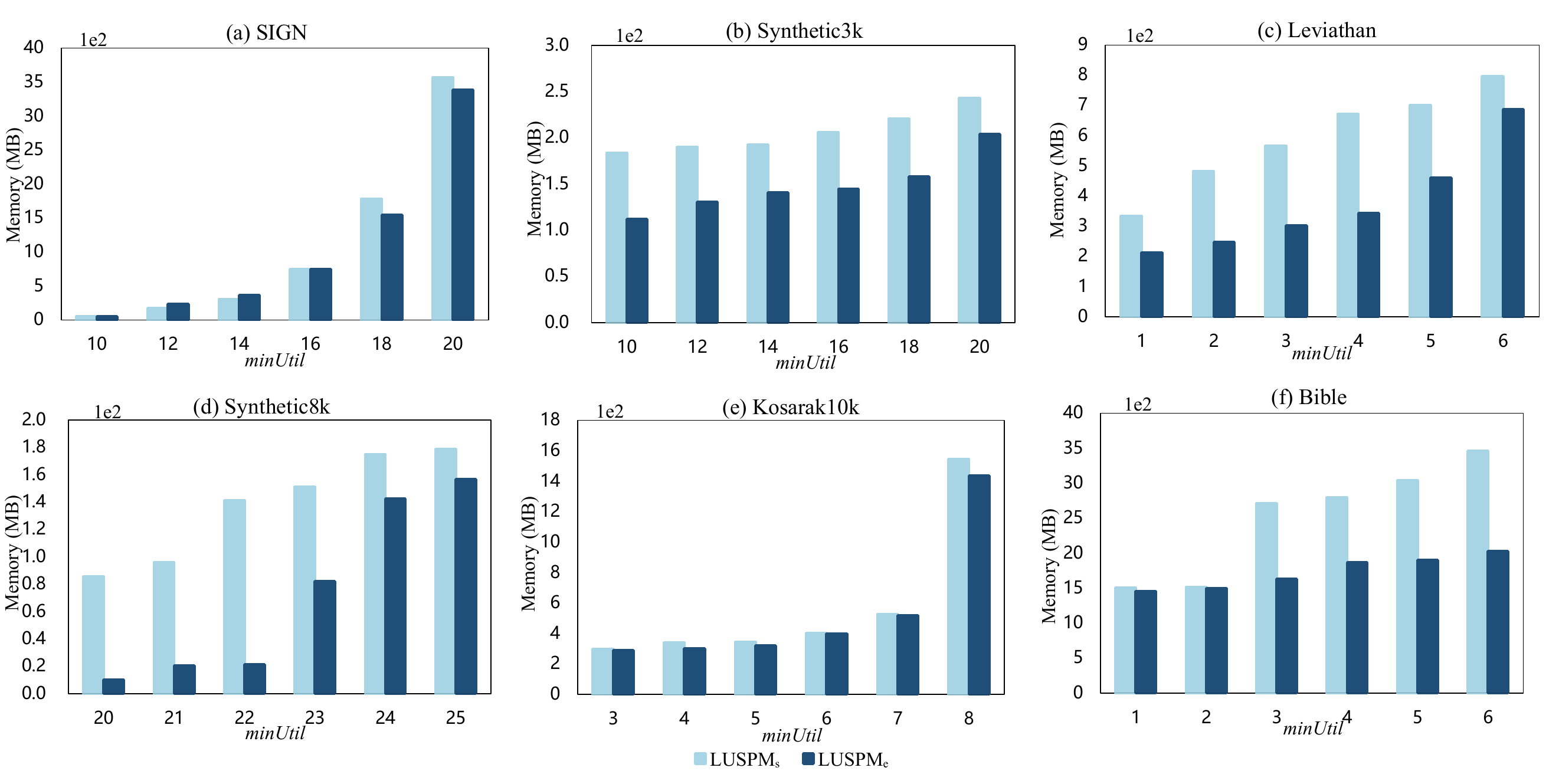}
    \caption{Memory consumption analysis of LUSPM$_s$ and LUSPM$_e$.}
    \label{fig:LUSPM Memory Consumpion Analysis}
\end{figure*}

\subsubsection{Performance Analysis of LUSPM$_b$}

In our experiment, LUSPM$_b$ failed to complete on the full datasets within two days. The most likely reason is that it relies on exhaustive enumeration to generate sequences and compute their utilities without employing any pruning strategies, resulting in excessive runtime. To further analyze its performance, we designed an additional experiment. Specifically, we tested LUSPM$_b$ on a single sequence from SIGN ($|S_1|$ = 44) and progressively increased its length from 20 to 34 items, in increments of 2. In other words, the algorithm was executed on sequences of 20, 22, 24, 26, 28, 30, 32, and 34 items. The corresponding runtimes were 3.754s, 14.949s, 62.726s, 262.677s, 1,068.073s, 4,362.813s, 16,819.089s, and 74,092.667s, respectively. As can be observed, the runtime grew exponentially with sequence length, nearly quadrupling with every two additional items. Ultimately, processing a 34-item sequence required nearly 20 hours. These results demonstrate that exhaustive enumeration is computationally impractical and highlight the necessity of pruning strategies to achieve acceptable performance.

\subsubsection{Performance Analysis of LUSPM$_{s}$ \& LUSPM$_{e}$}

We then evaluate the runtime, memory usage, and number of utility computations of LUSPM$_s$ and LUSPM$_e$ on six datasets under various \textit{minUtil} values without length constraints. Since the proposed algorithms are designed to discover LUSPs, the \textit{minUtil} parameter should be set to a sufficiently small value, representing only a very small proportion of the total database utility. Following low-utility itemset mining studies \cite{zhang2025enabling}, where \textit{minUtil} is typically set between $10^{-7}$ and $10^{-6}$ of the total database utility. Since the low-utility sequential pattern mining problem is more complex, we vary \textit{minUtil} from $10^{-8}$ to $10^{-5}$ of the total database utility to keep the runtime within a reasonable range.

\textbf{Runtime Evaluation:} Fig.~\ref{fig:LUSPM Time Consumption Analysis} shows the runtime of LUSPM$_s$ and LUSPM$_e$  on six datasets. Both algorithms can complete within a reasonable time, demonstrating significantly better runtime performance than LUSPM$_b$. Moreover, LUSPM$_e$ consistently outperforms LUSPM$_s$ across all datasets. For example, in the Synthetic3k dataset, when \textit{minUtil} = 20, the runtime of LUSPM$_s$ is approximately 7975s, whereas LUSPM$_e$ requires only 3406s, representing a reduction of about 57.3\%. In the Leviathan dataset, when \textit{minUtil} = 6, the runtime of LUSPM$_s$ is approximately 56319s, while LUSPM$_e$ requires 19650s, representing a reduction of about 65.1\%. This is probably because the pruning strategies in LUSPM$_e$ are more effective than those in LUSPM$_s$.
$\indent$  

\begin{figure*}[t!]
    \centering
    \includegraphics[width=0.85\textwidth]{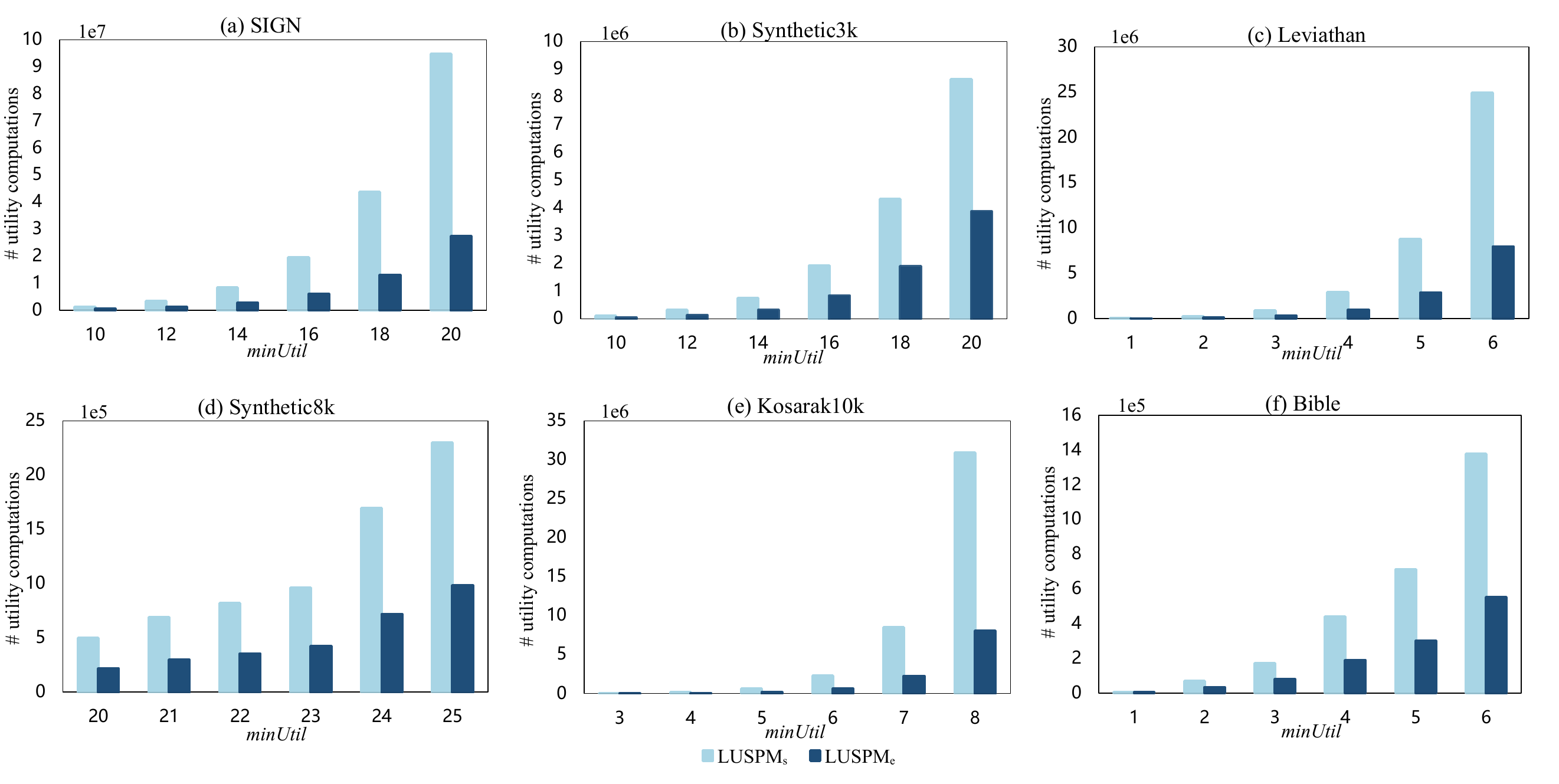}
    \caption{Number of utility computations of LUSPM$_{s}$ and LUSPM$_{e}$.}
    \label{fig:Number of utility computations}
\end{figure*}

\begin{figure*}[t!]
    \centering
    \includegraphics[width=0.85\textwidth]{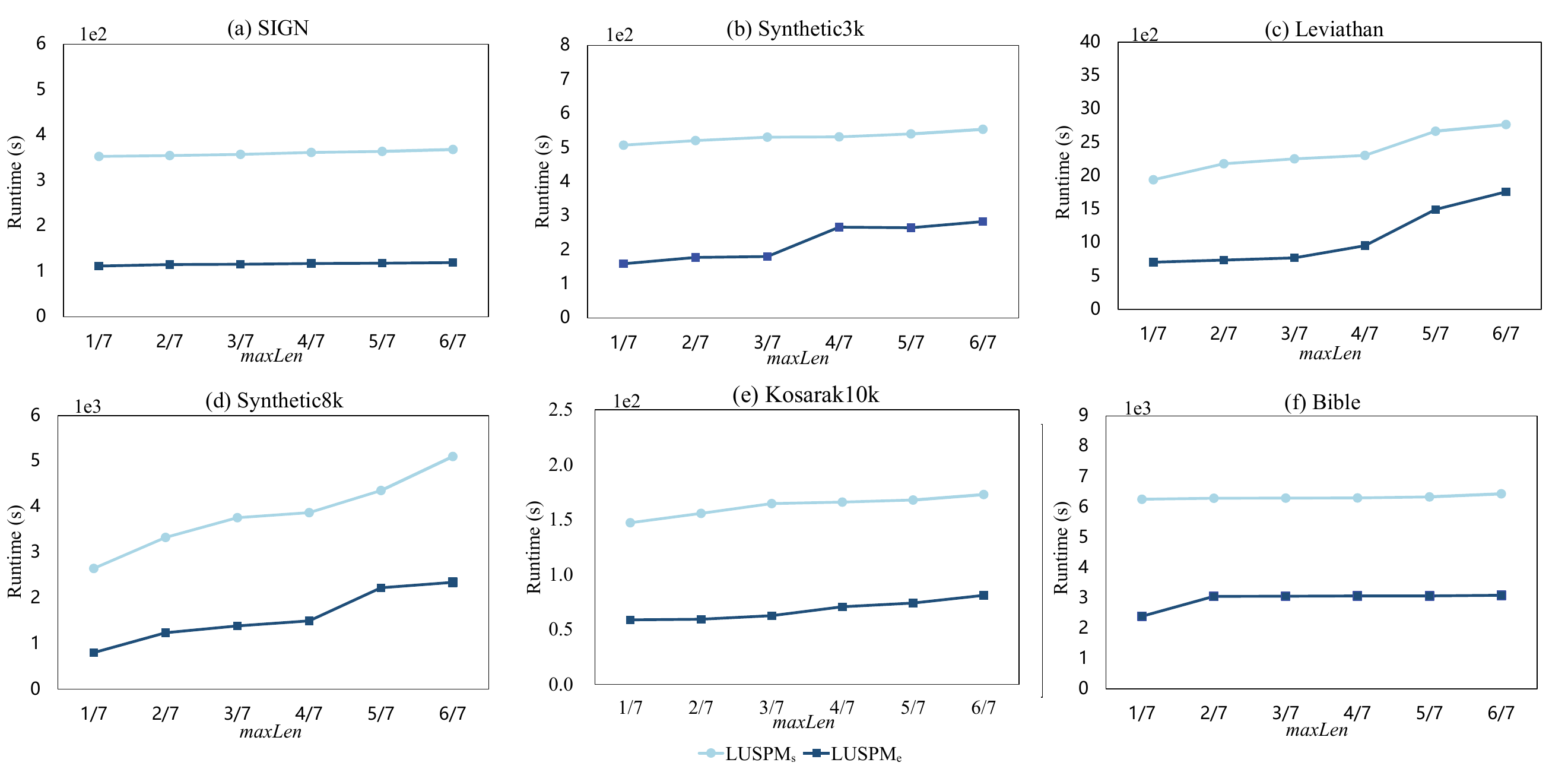}
    \caption{Runtime performance under different values of \textit{maxLen}.}
    \label{fig:LUSPM_length_time}
\end{figure*}

\textbf{Memory Evaluation:} We then compared the memory usage of the two algorithms. Fig.~\ref{fig:LUSPM Memory Consumpion Analysis} illustrates their performance across all the datasets. LUSPM$_e$ generally consumes slightly less memory than LUSPM$_s$ in most datasets. For example, in the Bible dataset, when \textit{minUtil} = 6, LUSPM$_s$ consumes approximately 3456 MB, whereas LUSPM$_e$ uses 2027 MB, representing a reduction of about 41.3\%. In the Leviathan dataset, when \textit{minUtil} = 6, LUSPM$_s$ consumes around 795 MB, while LUSPM$_e$ requires 686 MB, representing a reduction of about 13.7\%. In the SIGN dataset, when \textit{minUtil} = 20, LUSPM$_s$ consumes approximately 3568 MB, whereas LUSPM$_e$ uses 3386 MB, representing a reduction of about 5.1\%. This is probably because, although both algorithms rely on the same data structures, e.g., bit matrix, the sequence-utility chain, and MaxNonConSeqSet, the more effective pruning strategies in LUSPM$_e$ generally result in lower memory consumption.

\textbf{Utility Computations:} Fig.~\ref{fig:Number of utility computations} shows the number of utility computations for the two algorithms across all datasets. It is evident that LUSPM$_e$ consistently requires significantly fewer utility computations than LUSPM$_s$ on all datasets. For example, in Synthetic8k, when \textit{minUtil} = 25, LUSPM$_s$ performs 2,300,311 utility computations, whereas LUSPM$_e$ performs 981,329, representing a reduction of approximately 57.3\%. In Kosarak10k, when \textit{minUtil} = 8, LUSPM$_s$ performs 30,920,080 utility computations, while LUSPM$_e$ performs 8,003,661, representing a reduction of approximately 74.1\%. This is probably because the pruning strategy \ref{strategy 4} in LUSPM$_e$ significantly reduces the number of utility computations.

\subsection{Performance Under Different \textit{maxLen}s}

To further evaluate the proposed algorithms, we test LUSPM$_s$ and LUSPM$_e$ under a fixed \textit{minUtil} and varying \textit{maxLen}s. The \textit{minUtil} is set to the lower median from previous tests, corresponding to utility values of 14, 14, 3, 22, 5, and 3 for the six datasets, respectively, while \textit{maxLen} ranges from 1/7 to 6/7 of the maximum sequence length in each dataset.

\begin{figure*}[t!]
    \centering
    \includegraphics[width=0.85\textwidth]{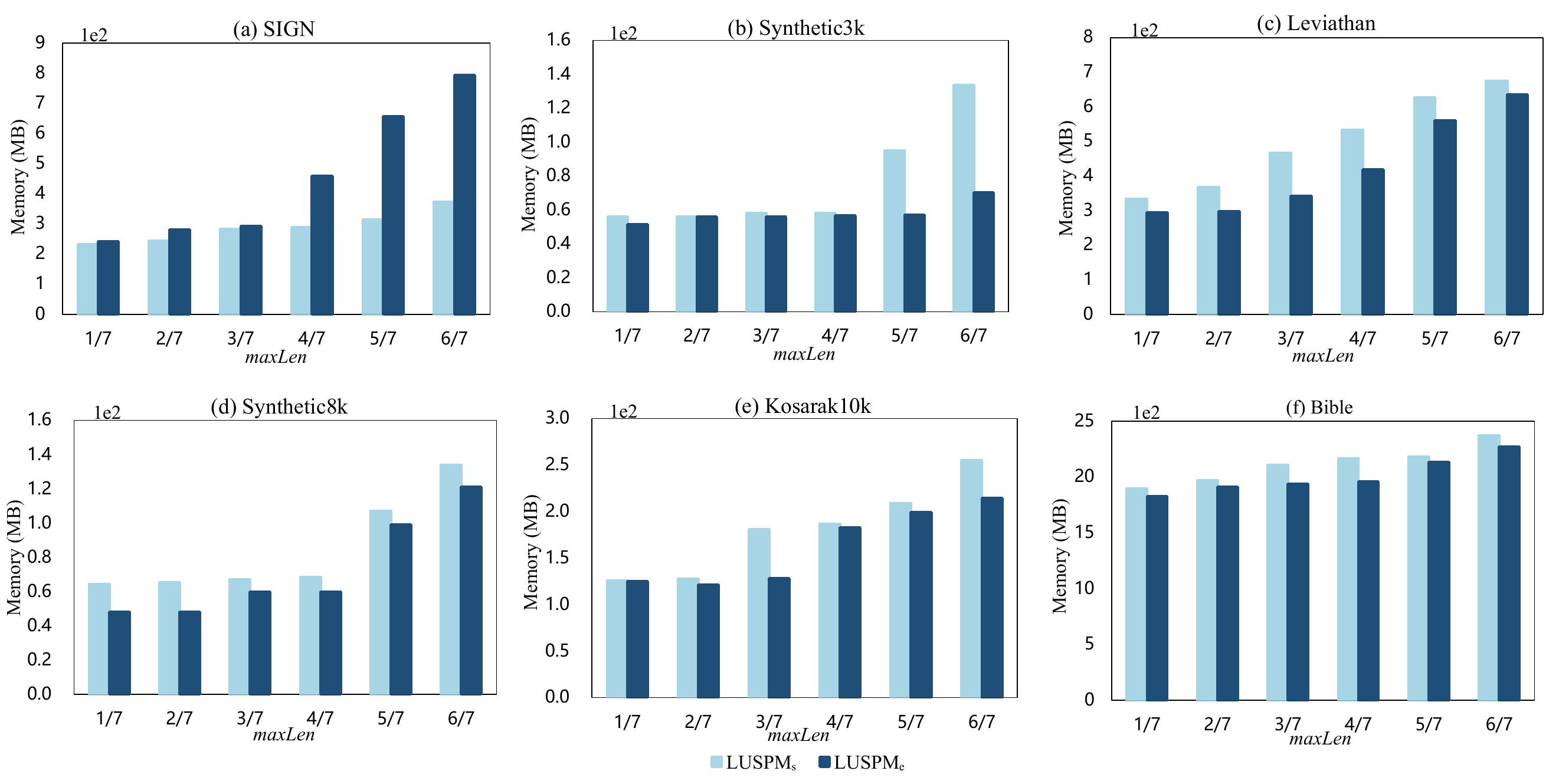}
    \caption{Memory performance under different values of \textit{maxLen}.}
    \label{fig:LUSPM_length_memory}
\end{figure*}

\textbf{Runtime Evaluation:} Fig.~\ref{fig:LUSPM_length_time} shows the runtime of the two algorithms on all datasets under various maximum length constraints. LUSPM$_e$ consistently outperforms LUSPM$_s$ across all datasets, consistent with the previous results obtained without length constraints, indicating that the pruning strategies in LUSPM$_e$ are more effective. Moreover, as the maximum sequence length increases, the runtime of both algorithms changes approximately linearly. Compared to the exponential growth observed in Fig. \ref{fig:LUSPM Time Consumption Analysis}, this change is relatively minor. The limited variation in runtime is due to Strategy \ref{strategy 1} in both algorithms, which effectively prunes many invalid items during preprocessing, thereby reducing the effective sequence length. Additionally, Fig~\ref{fig:LUSPM_length_time} shows that the runtime variation of LUSPM$_s$ is smaller than that of LUSPM$_e$. This is probably because Strategy \ref{strategy 3} in LUSPM$_s$ also efficiently prunes invalid items.

\textbf{Memory Evaluation:} Fig.~\ref{fig:LUSPM_length_memory} shows the memory consumption of the two algorithms under different maximum length constraints. Overall, LUSPM$_e$ generally consumes less memory than LUSPM$_s$. For example, in the Synthetic3k dataset, when maxLen = 6/7, LUSPM$_s$ consumes 133 MB, while LUSPM$_e$ consumes 70 MB, representing a reduction of approximately 47.3\%. This trend is consistent with the results obtained without length constraints, indicating that the pruning strategies in LUSPM$_e$ remain more effective across most datasets even as the maximum length varies. However, on the SIGN dataset, LUSPM$_s$ consumes slightly less memory than LUSPM$_e$. We speculate that this is because, under a \textit{minutil} of 14, the pruning strategies in LUSPM$_s$ are more effective for this dataset, possibly due to its sequence characteristics, which allow more items to be pruned efficiently.

\subsection{Scalability Analysis}

To assess scalability, we measured runtime and memory usage of LUSPM$_s$ and LUSPM$_{e}$ across varying dataset scales with \textit{minUtil} = 5 and no length constraint. We generate synthetic datasets with sequences of varying lengths (ranging from 50\textit{K} to 100\textit{K}) by randomly sampling rows from the six datasets in Table \ref{Tab:data} as well as from the YooChoose dataset\footnote{\url{https://archive.ics.uci.edu/dataset/352/online+retail}}. Fig. \ref{fig:Scalability Analysis} shows that both algorithms scale effectively on large datasets. Runtime increased with data size, with LUSPM$_{e}$ being faster than LUSPM$_s$, consistent with earlier results. Memory usage also increases but stabilizes when the dataset size exceeds 70\textit{K}, with LUSPM$_{e}$ maintaining a slight advantage. These results demonstrate that they are scalable to large-scale sequence datasets, making them suitable for real-world applications.

\begin{figure}[htbp]
    \centering
    \includegraphics[width=\linewidth]{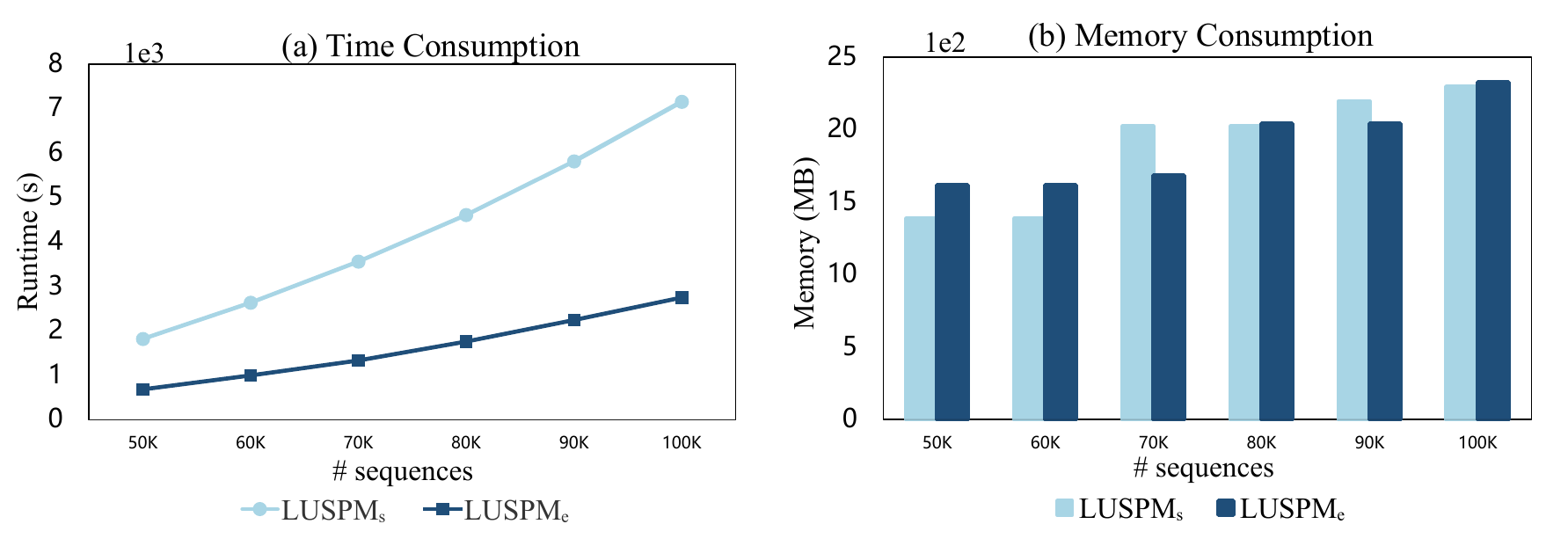}
    \caption{Scalability analysis of LUSPM$_{s}$ and LUSPM$_{e}$}
    \label{fig:Scalability Analysis}
\end{figure}

\section{Conclusion and Future work} \label{sec: conclusion}

In this paper, we first formalize the task of low-utility sequential pattern mining (LUSPM) and demonstrate that existing HUSPM algorithms are infeasible for discovering low-utility sequences. We then redefine sequence utility to capture the total utility and introduce the sequence-utility chain for efficient storage. Subsequently, we propose a baseline algorithm, LUSPM$_b$, to discover the complete set of low-utility sequential patterns. To reduce redundant processing, we further introduce the maximal non-mutually contained sequence set, along with a pruning strategy (Strategy \ref{strategy 1}). Building on this foundation, we propose two enhanced algorithms, LUSPM$_s$ and LUSPM$_e$, which significantly improve mining efficiency through three additional pruning strategies: Strategies \ref{strategy 2} and \ref{strategy 3} for LUSPM$_s$, and Strategy \ref{strategy 4} for LUSPM$_e$. Finally, extensive experiments demonstrate that both LUSPM$_s$ and LUSPM$_e$ substantially outperform LUSPM$_b$, with LUSPM$_e$ achieving the best runtime and memory efficiency while maintaining strong scalability.

Despite its contributions, this work has several limitations. The current framework assumes single-item events, whereas real-world data often involves itemsets. We plan to extend the framework to support itemset-based sequences, which would enhance its applicability and generalization. Additionally, the proposed methods focus on static datasets. In the future, we aim to explore incremental and streaming LUSPM to better handle evolving data.

\bibliographystyle{IEEEtran}
\bibliography{LUSPM.bib}

\end{document}